\documentclass{svjour3}                     % onecolumn (standard format)
\smartqed  % flush right qed marks, e.g. at end of proof
\usepackage{graphicx}
\usepackage{amssymb,epsfig,subfigure}
\usepackage{color}

\usepackage{amscd}

\begin{document}

\title{An instability criterion for nonlinear standing waves on nonzero backgrounds %\thanks{Grants or other notes
%about the article that should go on the front page should be
%placed here. General acknowledgments should be placed at the end of the article.}
}
%\subtitle{Do you have a subtitle?\\ If so, write it here}

%\titlerunning{Short form of title}        % if too long for running head

\author{R.K.\ Jackson  \and R.\ Marangell         \and
        H.\ Susanto}

%\authorrunning{Short form of author list} % if too long for running head

\institute{
           R.K.\ Jackson  \at
%              U. S. Naval Academy, Mathematics Department, 572 Holloway Road,\\ Annapolis, MD 21402-5002\\
%              \email{rkjackso@usna.edu}
              Hussman Strategic Advisors, \\ 
              5136 Dorsey Hall Drive, Ellicott City, MD 21042, USA\\
              \email{rjackson@hussmanfunds.com}
              \and
           R.\ Marangell \at
              Department of Mathematics and Statistics, \\The University of Sydney, Sydney 2006 NSW, Australia\\
              \email{robert.marangell@sydney.edu.au}           %  \\
%             \emph{Present address:} of F. Author  %  if needed
              \and
           H.\ Susanto \at
              School of Mathematical Sciences, University of Nottingham, University Park,\\ Nottingham NG7 2RD, UK\\
              \email{hadi.susanto@nottingham.ac.uk}
}

\date{\today}
% The correct dates will be entered by the editor

\maketitle

\begin{abstract}

A nonlinear Schr\"odinger equation with repulsive (defocusing) nonlinearity is considered. As an example, a system with a spatially varying coefficient of the nonlinear term is studied. The nonlinearity is chosen to be repelling except on a finite interval. Localized standing wave solutions on a non-zero background, \emph{e.g.}, dark solitons trapped by the inhomogeneity, are identified and studied. A novel instability criterion for such states is established through a topological argument. This allows instability to be determined quickly in many cases by considering simple geometric properties of the standing waves as viewed in the composite phase plane. Numerical calculations accompany the analytical results.

% \PACS{PACS code1 \and PACS code2 \and more}
% \subclass{MSC code1 \and MSC code2 \and more}
\end{abstract}

\section{Introduction}
\label{intro}

Dark solitons provide an important example of solitary waves on a nonzero background.  A typical dark soliton is a one-dimensional localized wavepacket consisting of a dip in an ambient nonzero density, where the carrier plane wave has a phase shift along the spatial variable. Such a solution is supported in nonlinear Schr\"odinger equations with repulsive or defocusing type nonlinearity, which admit stable nonzero backgrounds, describing localized excitations of, e.g., wavefunctions in nonlinear optics \cite{kivs98} and Bose-Einstein condensates \cite{kevr07}.

The dynamics of dark solitons in the nonlinear Schr\"odinger equation can be understood completely since the equation is exactly integrable via the inverse scattering technique \cite{zakh73}. When integrability is lost, an analysis of the existence and stability of dark solitons becomes more involved. A criterion relating the soliton linear (in)stability and the slope of the corresponding normalized momentum-velocity curve was derived in \cite{kivs95,bara96}, similar to the popular Vakhitov-Kolokolov condition for bright solitons \cite{vakh73,kolo73}. The method was later rigorously justified in \cite{lin02,peli08}. A Vakhitov-Kolokolov function whose sign determines the stability of the corresponding dark soliton was defined in \cite{menz07}, complementing the result of \cite{lin02} in the case of zero velocity. The Cauchy problem for the defocusing Schr\"odinger equation was shown to be locally well-posed and the stability of dark solitons was justified in \cite{zhid92}.

%In the presence of perturbations, dark solitons may become unstable.
Recent successful experiments on the generations of spatial dark solitons in nonlinear optics \cite{swar91,mora01} and particularly Bose-Einstein condensates \cite{burg99,dens00} (see also review \cite{fran10}) have initiated an intense theoretical study on the stability of dark solitons in the presence of various experimentally-relevant external potentials, which may destabilize the solitons. The (in)stability has been analyzed using a variety of analytical methods, such as asymptotic analysis \cite{peli96,lega97}, Evans function techniques \cite{kapi00,peli08}, particle-like approximations \cite{kivs94,kivs95_b}, perturbation methods based on Jost functions \cite{kono94,chen98,huan99,ao05} and the inverse scattering transform method \cite{lash04}.

%The literature on dark solitons is considerably more limited than their counterparts. In the context of nonlinear optics, this is possibly due to the required high powers for their realizations in experiments. As in the context of, e.g., Bose-Einstein condensates, such a limitation is not present, where many number of atoms can be loaded in a magnetic trap, it becomes relevant to consider dark solitons even experimentally.

Here, we introduce a novel topological method to recognize instabilities of solitons on nonzero backgrounds caused by purely real eigenvalues. The result enables us to study the instability of a dark soliton through a phase portrait in the phase-plane. The geometric analysis complements previous methods developed for bright solitons \cite{jone88} (see also \cite{mara11} and references therein) as well as the eigenvalue counts for dark solitons derived in \cite{peli08}. As an illustrative example, we consider solitary waves of a Schr\"odinger equation whose nonlinearity is a combination of self-focusing and self-defocusing, with the self-focusing region localized in a finite interval. Such a structure was first proposed and studied in \cite{chen92}. By introducing a non-repulsive inhomogeneity, it was shown that stable static gray solitons not present in the homogeneous case, which we refer to as the ``dips'' herein, are possible \cite{tran94}. Recently, such ``dips'' in motion were shown in \cite{abdu12} to have a critical velocity at which they disappear in a saddle-node bifurcation (see also \cite{haki97,mari03} for the existence and stability of moving ``dips'' due to linear inhomogeneity). 

In the experiments of Bose-Einstein condensates, Feshbach resonances that can control the sign and magnitude of the scattering length \cite{kohl06} were recently used to induce spatial inhomogeneities in the nonlinearity coefficient \cite{yama10}. The reader is also referred to a recent comprehensive survey of results for complex nonlinear wave patterns supported by nonlinear lattices and their combinations with linear lattices as well as relevant experimental settings in \cite{kart11}. Nevertheless, rather than considering infinitely long and periodic (linear and nonlinear) lattices, we study localised inhomogeneities tantamount to finite optical lattices \cite{caru02}.

This paper is outlined as follows. In Section 2 we discuss the mathematical model and describe its standing wave solutions.  In Section 3, we consider the corresponding linear eigenvalue problem, developing and applying an instability criterion for these standing waves. In Section 4, we compare these theoretical results with numerical computations. In Section 5, we consider the existence and stability of black solitons. In the section, we demonstrate that the result cannot be used to detect instability caused by complex eigenvalues. A few concluding remarks are included in Section 6.

\section{Mathematical model and standing waves}

We consider the following governing differential equation
\begin{equation}
\mathrm{i}\Psi_{t}+\Psi_{xx}+f(|\Psi|^2,x)\Psi=0,
\label{gov1}
\end{equation}
where
$$
f(|\Psi|^2,x) = \left\{\begin{array}{ll}
|\Psi|^2-V \qquad & \textrm{if $|x|<L$} \\
-|\Psi|^2 & \textrm{if $|x|>L$}
\end{array} \right. .
$$
In particular, this models a defocusing-type nonlinearity in the `outer' region and a focusing-type nonlinearity in the `inner' region. $L$ and $V$ are positive real-valued parameters modeling the inner system; $L$ is half the width of the inhomogeneity and $V$ is its depth. In the context of nonlinear optics or Bose-Einstein condensates, $V$ describes linear inhomogeneities due to variation in the refractive index or linear potential, respectively \cite{kart11}.

To study standing waves of (\ref{gov1}), we pass to a rotating
frame and consider solutions of the form $\Psi(x,t) = \psi(x,t) \ \mathrm{e}^{-\mathrm{i} r^2 t}$. We then have
\begin{equation}
\mathrm{i}\psi_{t}+\psi_{xx}+f(|\psi|^2,x)\psi=-r^2 \psi.
\label{gov2}
\end{equation}
Standing wave solutions of (\ref{gov1}) will be steady-state solutions of (\ref{gov2}). We consider real, $t$-independent solutions $u(x)$ of the ODE
\begin{equation}
u_{xx} = -(r^2+f(u^2,x))u.
\label{stat1}
\end{equation}
In the search for standing waves, we rewrite (\ref{stat1}) as a system
\begin{equation}
\begin{array}{rcl}
u_{x} & = & v \\
v_{x} & = & -(r^2+f(u^2,x))u
\end{array}
\label{stat2}
\end{equation}
and look for solutions in the phase plane that decay to a non-zero steady state as $x \rightarrow \pm \infty$. For each standing wave solution $u(x)$, we introduce the complementary power of the solution,
\[
N=\int_{-\infty}^{\infty}(u(x)^2-r^2)\,\mathrm{d}x.
\]
This value is not required in our analysis, but is a conserved quantity that we use to help depict the bifurcation diagrams in Figures \ref{fig:bif} and \ref{fig:truedark}.

\subsection{The outer system}
When $|x|>L$, system (\ref{stat2}) has the form
\begin{equation}
\begin{array}{rcl}
u_{x} & = & v \\
v_{x} & = & -r^2 u + u^3.
\end{array}
\label{stat2a}
\end{equation}
If $r^2 > 0$, the outer system (\ref{stat2a}) has three equilibria: $(-r, 0)$, $(0,0)$ and $(+r, 0)$.  $(0,0)$ is an elliptic fixed point. The twin equilibria $(\pm r, 0)$ are hyperbolic fixed points with identical linearizations; each has a one-dimensional local unstable manifold and a one-dimensional local stable manifold.

Further, note that this outer system (\ref{stat2a}) has Hamiltonian
$$
H_\textrm{out}(u,v)=\frac{1}{2}v^2+\frac{r^2}{2}u^2-\frac{1}{4}u^4.
$$
Solutions that approach the equilibria $(\pm r, 0)$ as $x \rightarrow \pm \infty$, must do so along the level set of the $H_\textrm{out}$ containing $(\pm r, 0)$, \textit{i.e.}, $H_\textrm{out} = r^4/4$.  In the $uv$-phase plane, this set consists of the parabolas
$$
v = \frac{\pm 1}{\sqrt{2}}\left(u^2-r^2\right).
$$
See the (red online) curves in Figure \ref{fig:pplane}.

\begin{figure}[tbhp!]
\begin{center}
\includegraphics[width=3.5in]{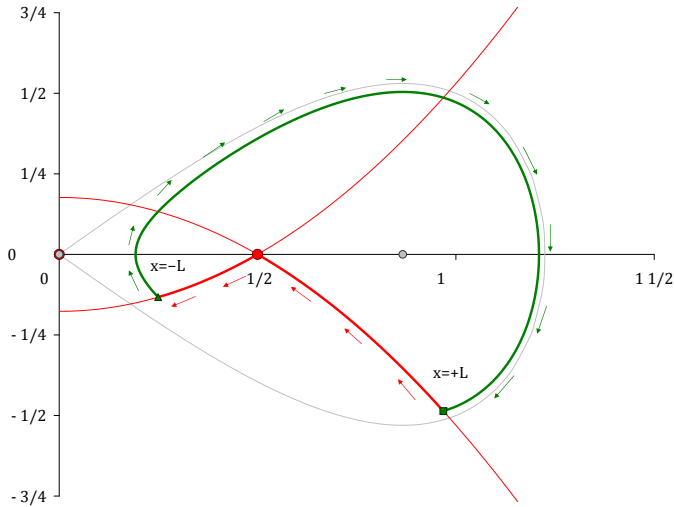}
\end{center}
\caption{A standing wave in the $uv$-plane with parameters $L=2.75$, $V=1$ and $r^2=0.25$.  For $x<-L$, solutions satisfy (\ref{stat2a}), and so any profile asymptotic to $(r,0)$ as $x \rightarrow -\infty$ must start along the unstable manifold of $(r,0)$, that is, the upward parabola.  At $x=-L$, the governing system switches to (\ref{stat2b}), and the solution leaves this parabola.  In order for the solution to also approach $(r,0)$ as $x \rightarrow +\infty$, the transient must reach a point on the stable manifold of $(r,0)$ (the downward parabola) at exactly $x=+L$, when the governing system switches back to (\ref{stat2a}).}
\label{fig:pplane}
\end{figure}

\subsection{The inner system}
When $|x|<L$, system (\ref{stat2}) takes the form
\begin{equation}
\begin{array}{rcl}
u_{x} & = & v \\
v_{x} & = & \left(V-r^2\right) u - u^3.
\end{array}
\label{stat2b}
\end{equation}
If $0<r^2<V$, this system also has three equilibrium points: $(-\sqrt{V-r^2},0)$, $(0,0)$ and $(+\sqrt{V-r^2},0)$. In this parameter regime, $(0,0)$ is a hyperbolic equilibrium and $(\pm \sqrt{V-r^2},0)$ are elliptic equilibria.  As $r^2$ is increased through $V$, these three equilibria collapse into one. And for $r^2>V$, there is a single elliptic equilibrium point at $(0,0)$.

In this inner system, trajectories travel along level curves of the Hamiltonian
$$
H_\textrm{in}(u,v)=\frac{1}{2}v^2+\frac{r^2-V}{2}u^2+\frac{1}{4}u^4.
$$
Portions of these level curves are depicted in green and gray in Figure \ref{fig:pplane} and elsewhere.

\subsection{Shooting for standing waves}

Standing waves can now be identified using a \emph{shooting argument} in the $uv$-phase plane.  For any standing wave that decays to $(r,0)$ as $x \rightarrow \pm \infty$, the trajectory $(u(x),v(x))$ must lie on the unstable manifold of $(r,0)$ in the outer system, that is, on the upward parabola for $-\infty < x < -L$.  At the other end, $(u(x),v(x))$ must be on the stable manifold of $(r,0)$ in the outer system, that is, on the downward parabola for $L < x < \infty$.  So a standing wave corresponds to a connection, evolving according to (\ref{stat2b}) for time $2L$, between a point $(u(-L),v(-L))$ on the upward parabola and a point $(u(L),v(L))$ on the downward parabola, see Figure \ref{fig:pplane}.

\begin{figure}[tb!]
\begin{center}
\includegraphics[width=3.5in]{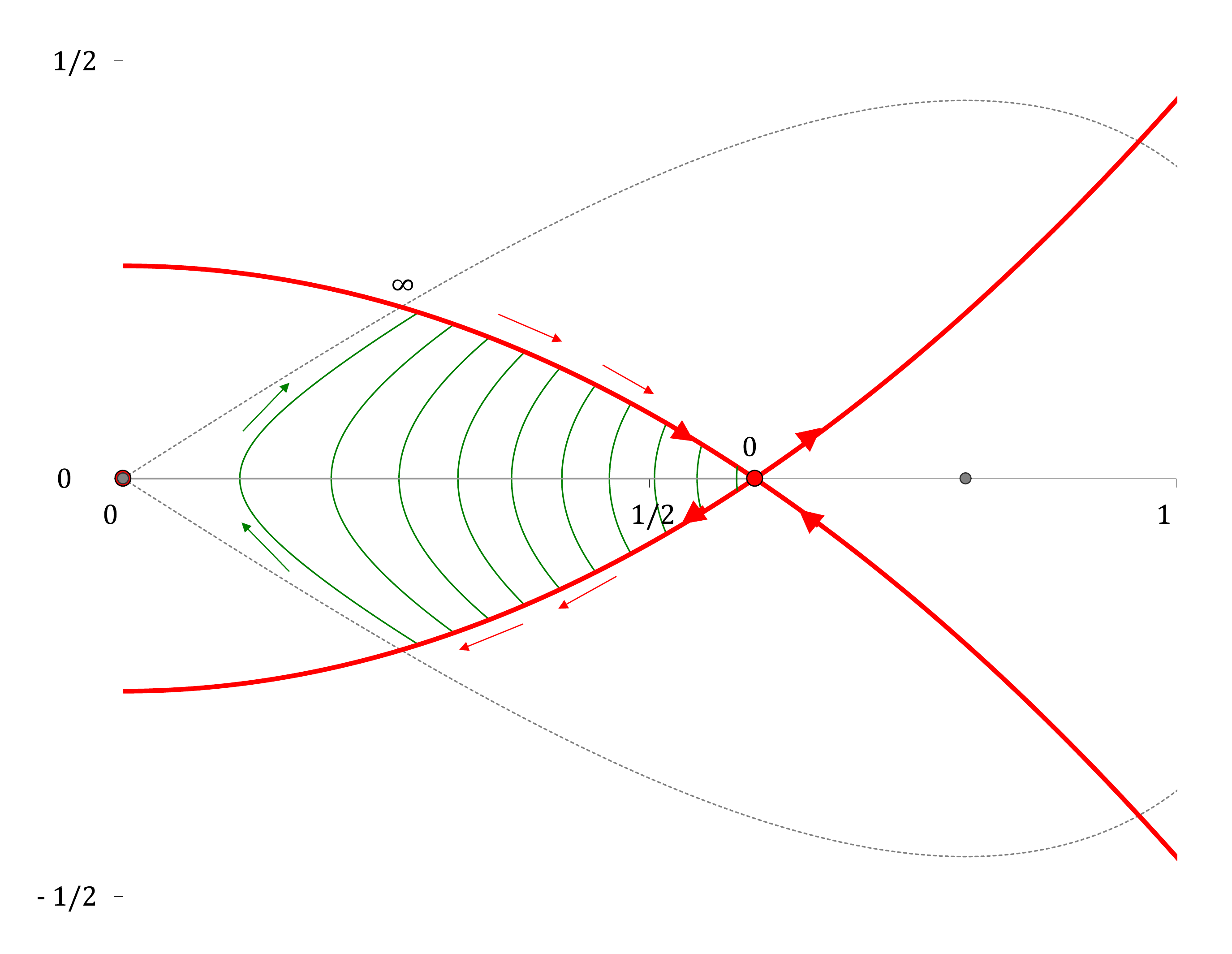}
\includegraphics[width=3.5in]{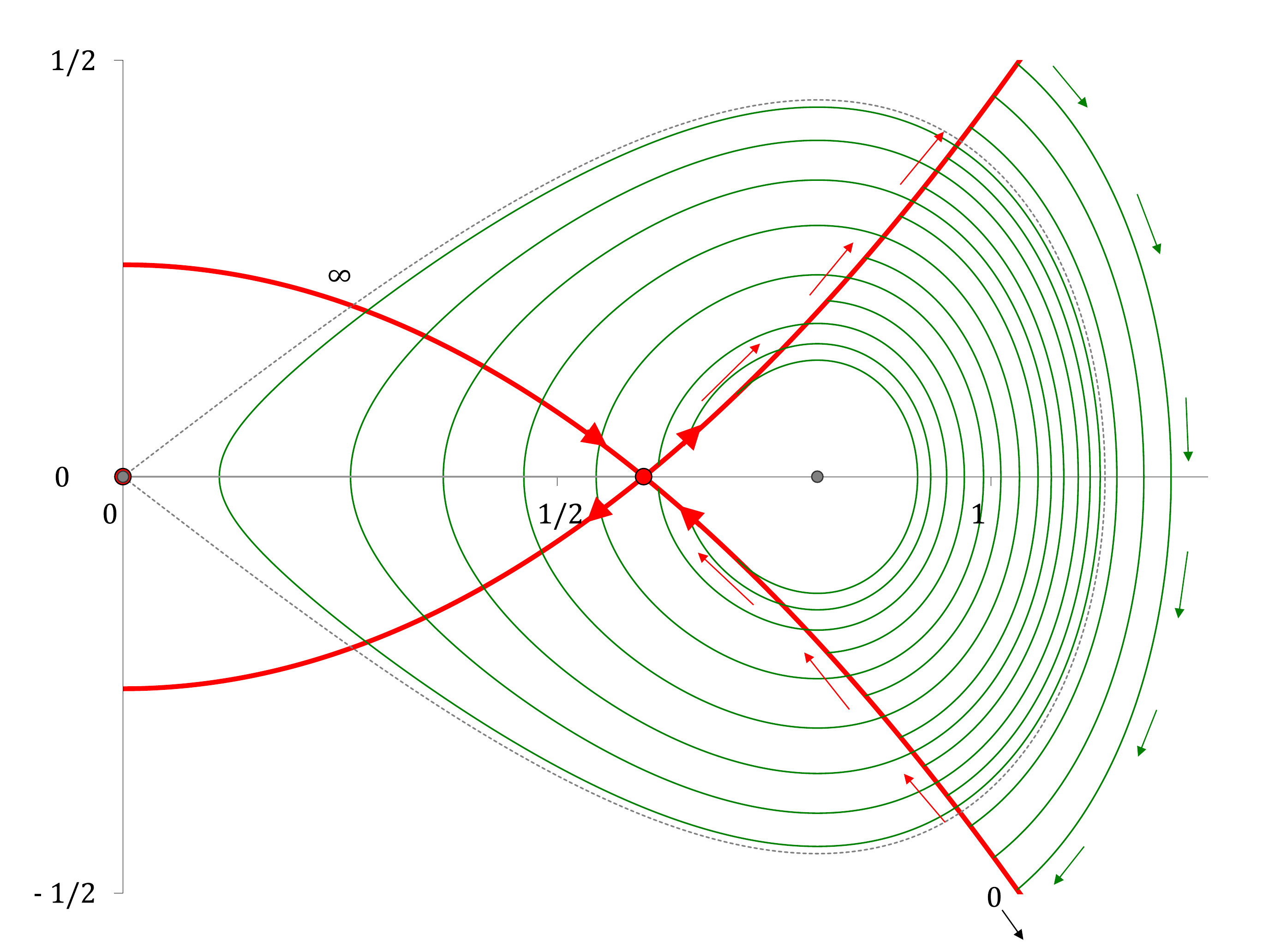}
\end{center}
\caption{(Top) For every $L$, there is a ``dip'' standing wave when $0 < r^2 < V/2$. The ``$0$'' and ``$\infty$'' along the downward parabola indicate the switching points where the time $2L$ is $0$ and $\infty$, respectively. (Bottom) For every $L$, there is also a ``bubble'' standing wave when $0 < r^2 < V/2$. In this case, the time $2L$ decreases when the switching point $u_0$ increases.}
\label{fig:dip}
\end{figure}

\begin{theorem}  \label{th:exist} Consider any positive system parameters $V$ and $L$.  Suppose the propagation parameter $r^2$ satisfies $0 < r^2 < V/2$.  Then there are at least \textbf{two} positive-valued symmetric standing waves, which we call a ``dip'' and a ``bubble''.  Additionally, if $L$ is sufficiently large, there is also a pair of positive-valued asymmetric standing waves.
\end{theorem}

\begin{proof}

We approach this proof a little sideways, first fixing a positive parameter $V$ and a propagation parameter $r^2$ in the interval $0<r^2<V$.  As noted above, a standing wave corresponds to a trajectory, evolving according to (\ref{stat2b}), linking the upward parabola $v = (u^2-r^2)/\sqrt{2}$ with the downward parabola $v = -(u^2-r^2)/\sqrt{2}$ in the $uv$-plane.  Instead of treating $L$ as fixed, we tune $L$ continuously, as a function of $u_0$, so that if $(u(-L),v(-L))= (u_0, (u_0^2-r^2)/\sqrt{2})$, that is, $(u(-L),v(-L))$ is on the upper (starting) parabola, then $(u(+L),v(+L))$ is on the downward (target) parabola.

First we identify the symmetric ``dip'' solution.  As starting values for our shot across the inner system, we choose $u_0$ satisfying $(V-\sqrt{V^2-2r^4})/2 < u_0^2 < r^2$.   Geometrically, this picks out potential starting points between the stable manifold of $(0,0)$ in the inner system and the equilibrium point $(r,0)$ in the outer system, see Figure \ref{fig:dip}.  For each $u_0$ in this interval, consider the time ($2L$) that it takes a trajectory, starting at $(u_0,(u_0^2-r^2)/\sqrt{2})$ on the upward parabola and evolving according to (\ref{stat2b}), to first reach the reflection $(u_0,-(u_0^2-r^2)/\sqrt{2})$.  Because $r^2 < V/2$, the inner system has $v_x(r,0) > 0$ and the transit time $2L$ between $(u_0,(u_0^2-r^2)/\sqrt{2})$ and $(u_0,-(u_0^2-r^2)/\sqrt{2})$ approaches zero as $u_0 \rightarrow r$ from the left.  At the other end, as the starting point $(u_0,(u_0^2-r^2)/\sqrt{2})$ approaches the stable manifold of $(0,0)$ in the inner system, and the connecting trajectory enters a neighborhood of the equilibrium $(0,0)$ of the inner system, the transit time approaches infinity.  Since the transit time depends continuously on the starting value $u_0$ in this interval, the Intermediate Value Theorem guarantees at least one ``dip'' solution for every positive value of $L$.

Second, we identify a symmetric ``bubble'' solution.  Here we consider starting values $u_0$ in the infinite half-interval $(V-\sqrt{V^2-2r^4})/2 < u_0^2$, see bottom panel of Figure \ref{fig:dip}.  As $u_0$ increases, the nonlinearity dominates and the time, $2L$, it takes for the trajectory starting at $(u_0,(u_0^2-r^2)/\sqrt{2})$ to first reach its reflection $(u_0,-(u_0^2-r^2)/\sqrt{2})$ approaches zero.  As we continue this trajectory by decreasing $u_0$, we near $u_0^2 = (V-\sqrt{V^2-2r^4})/2$ and again the transit time again approaches $\infty$.  Note that to deform this family of trajectories smoothly through $u_0=r$, we no longer cut off the trajectory at the time of ``first return'', but rather travel once around the loop before stopping. The Intermediate Value Theorem guarantees at least one ``bubble'' solution for every positive value of $L$.

\begin{figure}[tbhp!]
\begin{center}
%\subfigure
{\includegraphics[width=3.5in]{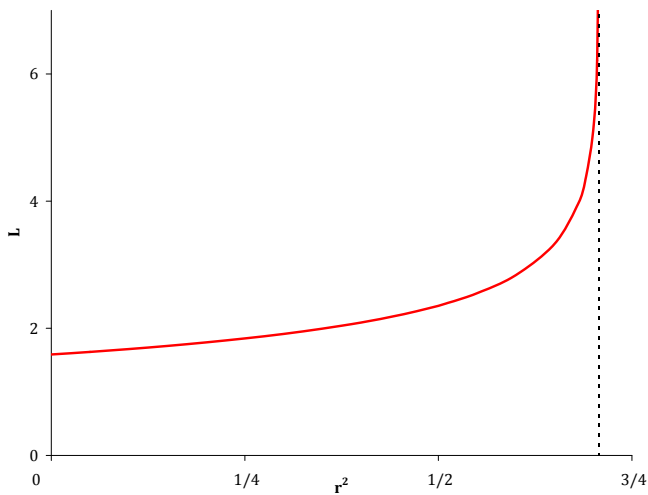}}
%\subfigure{\includegraphics[width=3.in]{bif1}}
%\subfigure{\includegraphics[width=3.in]{bif2}}
\end{center}
\caption{
The %top panel show the
threshold value $L_s$ for the appearance of an asymmetric solution as a function of $r^2$. Looked at another way, for each $L$, this graph identifies the $r^2-$location of the pitchfork bifurcation. %In the bottom panel, bifurcation diagrams ($N$ versus $r^2$) of positive standing waves with $V=1$ and two values of $L$ are shown. When $L=1$ only the two positive-valued symmetric pulses are present. For large enough $L$ (depending upon $V$ and $r^2$), there are asymmetric standing waves. When $L=2$, we see the asymmetric branch appearing. The numbers besides the curves show the value of $P-Q$ for the corresponding solutions (see the text).
}
\label{fig:L}
\end{figure}

Finally, for larger values of $L$, we identify a pair of asymmetric solutions.  Note especially that when $u_0 < \sqrt{V/2}$, the bubble solution identified above, starting at $(u_0,(u_0^2-r^2)/\sqrt{2})$, actually passes through the downward (target) parabola at a point $(\sqrt{V-u_0^2},-(V-u_0^2-r^2)/\sqrt{2})$ before terminating at the reflection $(u_0,-(u_0^2-r^2)/\sqrt{2})$.  So if $u_0$ satisfies $(V-\sqrt{V^2-2r^4})/2 < u_0^2 < V^2/2$, we can consider the trajectory of (\ref{stat2b}) that originates at $(u_0,(u_0^2-r^2)/\sqrt{2})$ and terminates when it first reaches $(\sqrt{V-u_0^2},-(V-u_0^2-r^2)/\sqrt{2})$.  Once again, as $u_0$ shrinks toward $\sqrt{(V-\sqrt{V^2-2r^4})/2}$, the transit time $2L$ approaches infinity.  On the other hand, as $u_0$ grows toward $\sqrt{V/2}$, the transit time approaches the transit time of the symmetric bubble that launches from $u_0 = \sqrt{V/2}$.  If this distinguished symmetric bubble has transit time $2L_s$, then we can guarantee an asymmetric bubble solution for all $L>L_s$. The threshold value $L_s$ varies with the propagation parameter $r^2$, as shown in Figure \ref{fig:L}.  For each $L>L_s$, the second (twin) asymmetric solution is a reflected version of the first, satisfying $u(-L)=\sqrt{V-u_0^2}$ and $u(+L)=u_0$. %In the figure, we also depict two bifurcation diagrams for the existence of positive solutions for $L=1$ and $2$. When $L=2$, the appearance of a branch from a pitchfork bifurcation corresponding to asymmetric solutions can be seen.
\qed
\end{proof}

Although the theorem above guarantees the existence of a number of standing waves when $r^2 < V/2$, it is an interesting study to consider the bifurcations that create and destroy these standing waves as parameters change.  In particular, when $r^2 > V/2$, the picture changes significantly.  For small values of $L$, the pulse that we have identified as the ``dip'' actually has the form of a small bubble ($u(x)>r$ for all $x$).  And for these small values of $L$, the ``bubble'' and the ``dip'' eventually meet and annihilate each other in a saddle-node bifurcation, see Figure \ref{fig:saddle_node}.

\begin{figure}[tbhp!]
\begin{center}
\includegraphics[width=3.5in]{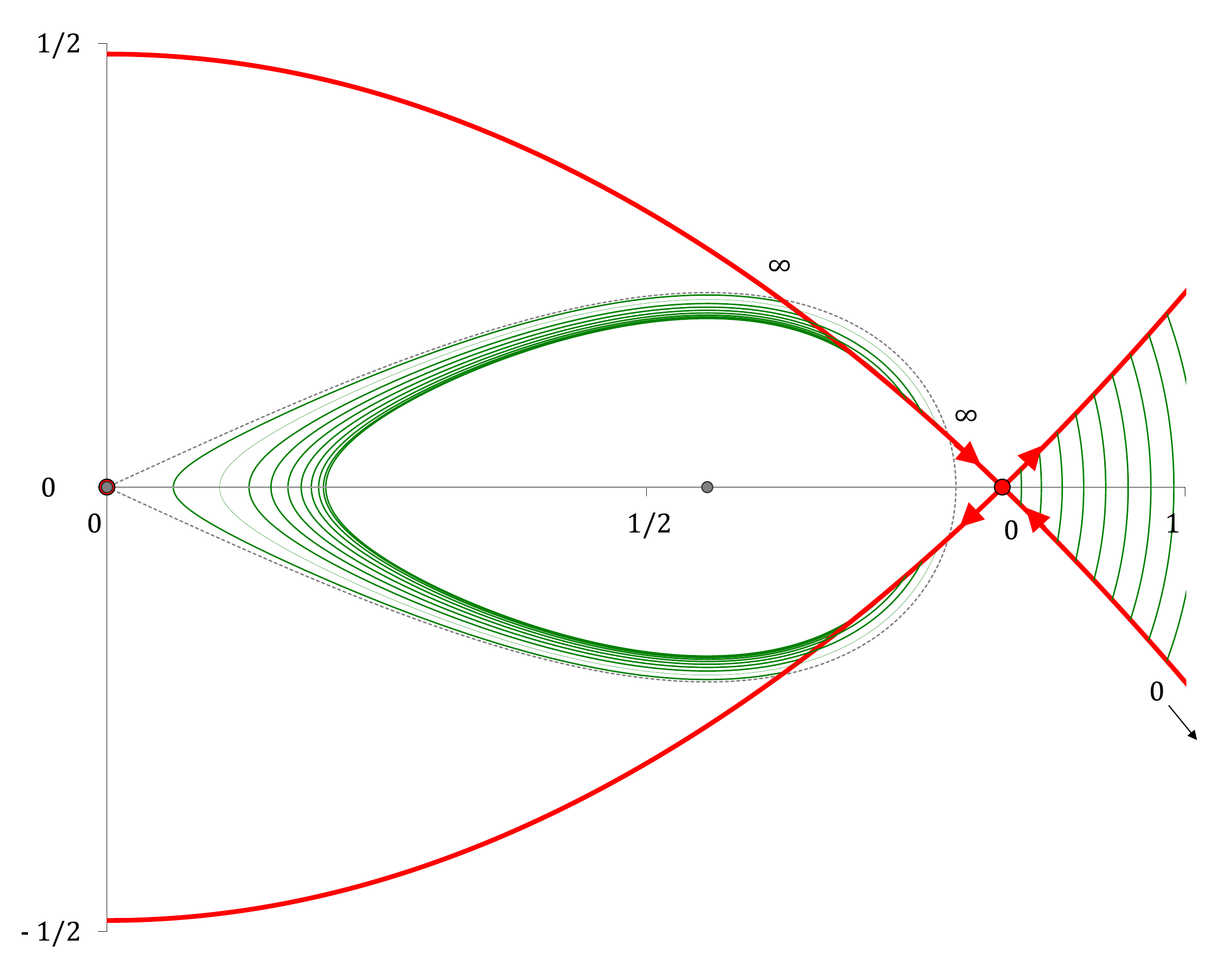}
\end{center}
\caption{For $r^2 > V/2$, pairs of waves are destroyed in saddle-node bifurcations.}
\label{fig:saddle_node}
\end{figure}

For larger values of $L$, there may be more than one bubble-type solution present.  Nonetheless, the ``dip'' and one of the bubbles similarly annihilate each other in a saddle-node bifurcation.  As $L \rightarrow \infty$, the bifurcation value for this saddle-node approaches $r^2=\sqrt{V/2}$.  Exactly which bubble meets to annihilate the dip changes with $L$. % (See Figure \ref{fig:diagrams}, where the bubbles swap partners between $L=4$ and $L=5$).

Asymmetric pulses are formed when a symmetric ``bubble'' undergoes a subcritical pitchfork bifurcation.  This occurs as the launch parameter falls through $u_0=\sqrt{V/2}$ with changing $r^2$.  At this value of $u_0$, the ``inner'' trajectory is tangent to the parabolas.

\section{Instability of standing waves}

To establish the instability of a standing wave solution, we linearize (\ref{gov2}) about a solution of (\ref{stat1}). Writing $\psi=u(x) + \epsilon\left((p(x)+\mathrm{i}q(x))\mathrm{e}^{\lambda t} + (p(x)^\star+\mathrm{i}q(x)^\star) \mathrm{e}^{\lambda^\star t}\right)$ and retaining terms linear in $\epsilon$ leads to the eigenvalue problem
\begin{equation}
\lambda\left(\begin{array}{cc} p \\ q \end{array}\right) = \left(\begin{array}{cc} 0 & D_{-} \\ - D_{+} & 0 \end{array}\right) \left(\begin{array}{cc} p \\ q \end{array}\right)  =: M \left(\begin{array}{cc} p \\ q \end{array}\right),
\label{eq:linear}
\end{equation}
where the linear operators $D_{+}$ and $D_{-}$ are defined as
\begin{eqnarray}
D_{+} & = & \frac{\partial^2}{\partial x^2} +r^2 + f(u^2,x)+ 2u^2D_1f(u^2,x),
\label{eq:plusoperator} \\
D_{-} & = & \frac{\partial^2}{\partial x^2} +r^2 + f(u^2,x).
\label{eq:minusoperator}
\end{eqnarray}
It is then clear that the presence of an eigenvalue of $M$ with positive real part implies instability.

\subsection{An instability criterion}

We will concern ourselves primarily with real eigenvalues, and we introduce a simple lower bound for the number of positive real eigenvalues of the operator $M$ defined in (\ref{eq:linear}). In \cite{rmckrtjhs10}, systems like (\ref{gov1}) were considered. One can show that the following quantities are well defined (see for example \cite{jone88}, and the references therein):
\begin{eqnarray*}
P &=& \textrm{ the number of positive eigenvalues of } D_{+} \\
Q &=& \textrm{ the number of positive eigenvalues of } D_{-}.
\end{eqnarray*}
From Sturm-Liouville theory, $P$ and $Q$ can be determined by considering nontrivial solutions of $D_+ v = 0$ and $D_- v = 0$, respectively. In fact, they are the number of zeros of the associated solution $v$. Notice that $D_- v = 0$ is actually satisfied by the standing wave itself, and that $D_+ v = 0$ is the equation of variations of the standing wave equation. It follows that:
\begin{equation} \label{eq:pandq}
\begin{array}{lll}
&&Q = \textrm{ the number of zeros of the standing wave $u$.}  \\
&&P = \textrm{ the number of zeros of a nontrivial solution to the variational equation along $u$. }
\end{array}
\end{equation}
Suppose that we are away from a bifurcation point, {so that $D_+ v = 0$ does not have a nontrivial bounded solution}.  Then using the definition of $P$ and $Q$, the main result of this paper is formulated in the following lemma.
\begin{lemma}\label{th:main}
If $P-Q \neq 0,$ then there is a real positive eigenvalue of the operator $M$.
\end{lemma}

\begin{proof}
Following the development of \cite{jone88}, we interpret the search for real eigenvalues as a shooting problem in the three-dimensional space of Lagrangian planes, denoted $\Lambda(2)$.  Considering (\ref{eq:linear}) as a first-order four-dimensional system, if $\lambda$ is real and positive, then the asymptotic system has four distinct real eigenvalues $\mu_{ss} < \mu_{s} < 0 < \mu_{u} < \mu_{uu}$ with eigenvectors $w_{ss}$, $w_{s}$, $w_u$, $w_{uu}$, respectively.  Such a positive real number $\lambda$ is an eigenvalue of (\ref{eq:linear}) if there is a nontrivial connection between the two-dimensional \emph{unstable} space at $-\infty$, spanned by $w_u$ and $w_{uu}$, and the \emph{train} of the two-dimensional \emph{stable} space at $+\infty$ spanned by $w_s$ and $w_{ss}$.  If $\lambda$ is not an eigenvalue, there is no connection, and the unstable subspace must connect back to itself, forming a closed loop in $\Lambda(2)$.

Note that the space of Lagrangian planes, $\Lambda(2)$, has a natural notion of winding known as the Maslov index.  The fundamental group of $\Lambda(2)$ is the integers, $\mathbb{Z}$, and the Maslov index of a closed curve in
$\Lambda(2)$ is the winding number of that curve.  And so, for any real number $\lambda$ that is not an eigenvalue, we can compute the Maslov index, denoted $I(\lambda)$, for the solution of (\ref{eq:linear}) asymptotic to the unstable space at $-\infty$.  If, for two real numbers $\lambda_1$ and $\lambda_2$, we find that $I(\lambda_1) \neq I(\lambda_2)$, then the closed loop must break for some value of $\lambda$ between $\lambda_1$ and $\lambda_2$.  At that point, there is a nontrivial connection between the \emph{unstable} space at $-\infty$ and the \emph{train} of the \emph{stable} space at $+\infty$, that is, $\lambda$ is an \emph{eigenvalue}.

For $\lambda$ sufficiently large there is no winding at all and we have $I(\lambda)=0$.  As $\lambda$ approaches zero, the equations for $p(x)$ and $q(x)$ in (\ref{eq:linear}) decouple and the dynamics reduce to the flow on two circles $S^1 \times S^1$ inside $\Lambda(2)$.  As in \cite{jone88}, it is the flow on these two circles within $\Lambda(2)$ that allows us to compute $I(\lambda)$ for $\lambda$ near zero.  However, in that previous work, the Maslov index could not be used directly because the trajectories of interest were not closed.  $0$ was necessarily an eigenvalue of (\ref{eq:linear}) since the standing wave $u(x)$ itself served as an eigenfunction of $D_{-}$ in (\ref{eq:minusoperator}).  In the present case, $0$ is generally \emph{not} an eigenvalue, as the standing wave solution $u(x)$ is not square-integrable.  In particular, away from bifurcation points where there is a nontrivial bounded solution of $D_{+}v = 0$, the distinguished orbit in $\Lambda(2)$ associated with $\lambda=0$ forms a closed loop.  We can then immediately conclude that $I(0) = P-Q$.  Hence, there must be at least $|P-Q|$ positive real eigenvalues (counted with multiplicity) and a standing wave can be stable only if $P-Q=0$. \qed
\end{proof}

\begin{remark}
\label{rem1}
Lemma \ref{th:main} complements the counting argument of \cite{peli08} (see Theorem 3.14 therein). Under some assumptions on %the operator $D_\pm$ and 
the algebraic simplicity of the eigenvalues of the spectral problem (\ref{eq:linear}), the theorem shows that if $N_c$ is the number of complex eigenvalues $\lambda$ in the first quadrant, $N^-_i$ is the number of purely imaginary eigenvalues $\lambda$ with Im$(\lambda) > 0$ and $(q,D_+^{-1}q)\leq0$ and $N_r^\pm$ is the number of real positive eigenvalues $\lambda$ with $(q,D_+^{-1}q)\geq0$ and $(q,D_+^{-1}q)\leq0$ respectively, where $(\cdot,\cdot)$ is the standard inner product in $L^2(\mathbb{R})$, then 
\begin{equation}
N_r^-+N_c+N_i^-=P,\quad N_r^++N_c+N_i^-=Q.
\label{ca}
\end{equation}
Using (\ref{ca}), we obtain that when $P=Q$ there will be at most $(P+Q)$ unstable eigenvalues. In particular, when $P=Q=0$, we know that the solution must be stable. 
\end{remark}

\subsection{Phase plane results}

We now use Lemma \ref{th:main} to establish the \emph{instability} of the bubble and asymmetric standing waves guaranteed in Theorem \ref{th:exist}.

\begin{theorem}  \label{th:stability} Consider the positive-valued standing waves described in Theorem \ref{th:exist}.  In parameter regimes where only a ``dip'' and a ``bubble'' are present, the bubble must be unstable with at least one positive real eigenvalue.  In parameter regimes where a ``dip'', a ``bubble'' and two asymmetric waves are present, the asymmetric waves are both unstable with at least one positive real eigenvalue, while the symmetric ``bubble'' is unstable with at least two positive real eigenvalues.
\end{theorem}

\begin{proof}
All of these pulses are positive-valued, so $Q=0$.  Successful application of the instability criterion depends upon the computation of $P$ for each of the standing waves under consideration.

$P$ can be determined by evolving a vector $p(x)$ according to (\ref{eq:plusoperator}), where $p(x)$ is initialized to be tangent to the unstable manifold of $(r,0)$ in the outer system (\ref{stat2a}), that is, on the upward parabola. $P$ is the number of times such a vector passes through verticality.  Recall that (\ref{eq:plusoperator}) acts as the equation of variations for the standing wave equation; except at the interfaces $x=\pm L$, vectors tangent to the standing wave in the $uv$-phase plane will remain tangent as they evolve.

\begin{figure}[tbhp!]
\begin{center}
%\subfigure{\includegraphics[width=4.0in]{robby_diagrams_a}}
\subfigure{\includegraphics[width=2.3in]{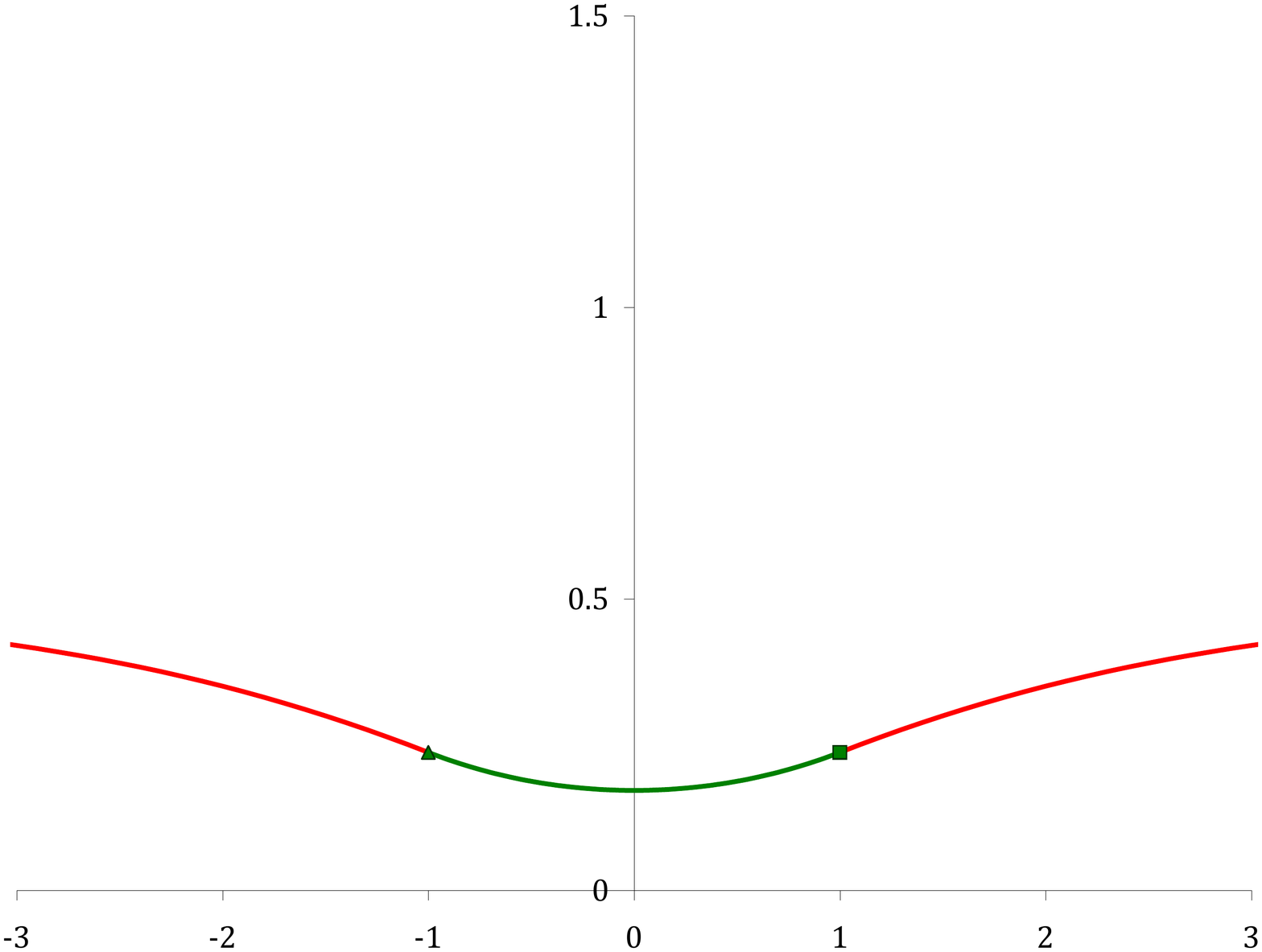}}
\subfigure{\includegraphics[width=2.3in]{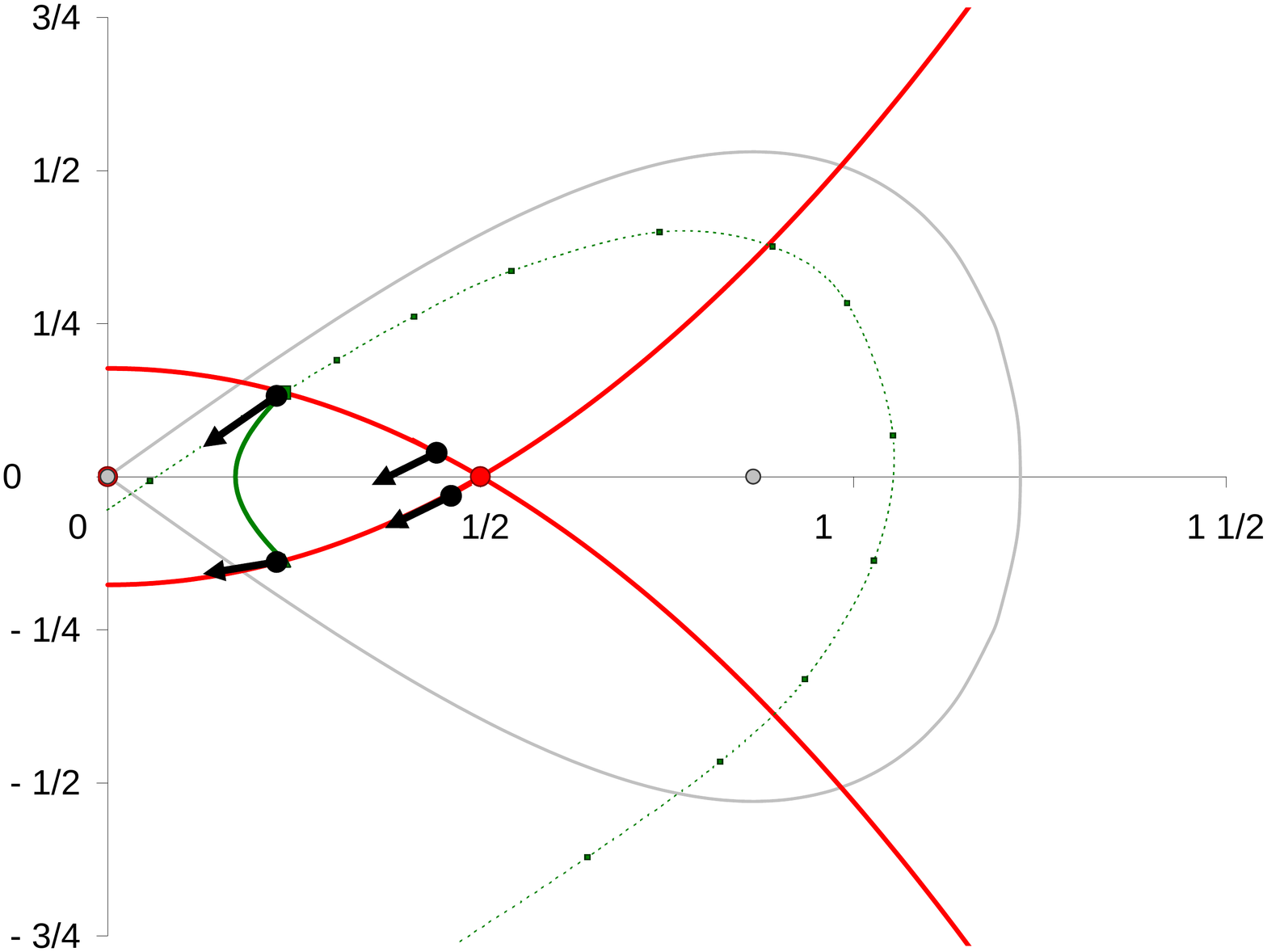}}
\subfigure{\includegraphics[width=2.3in]{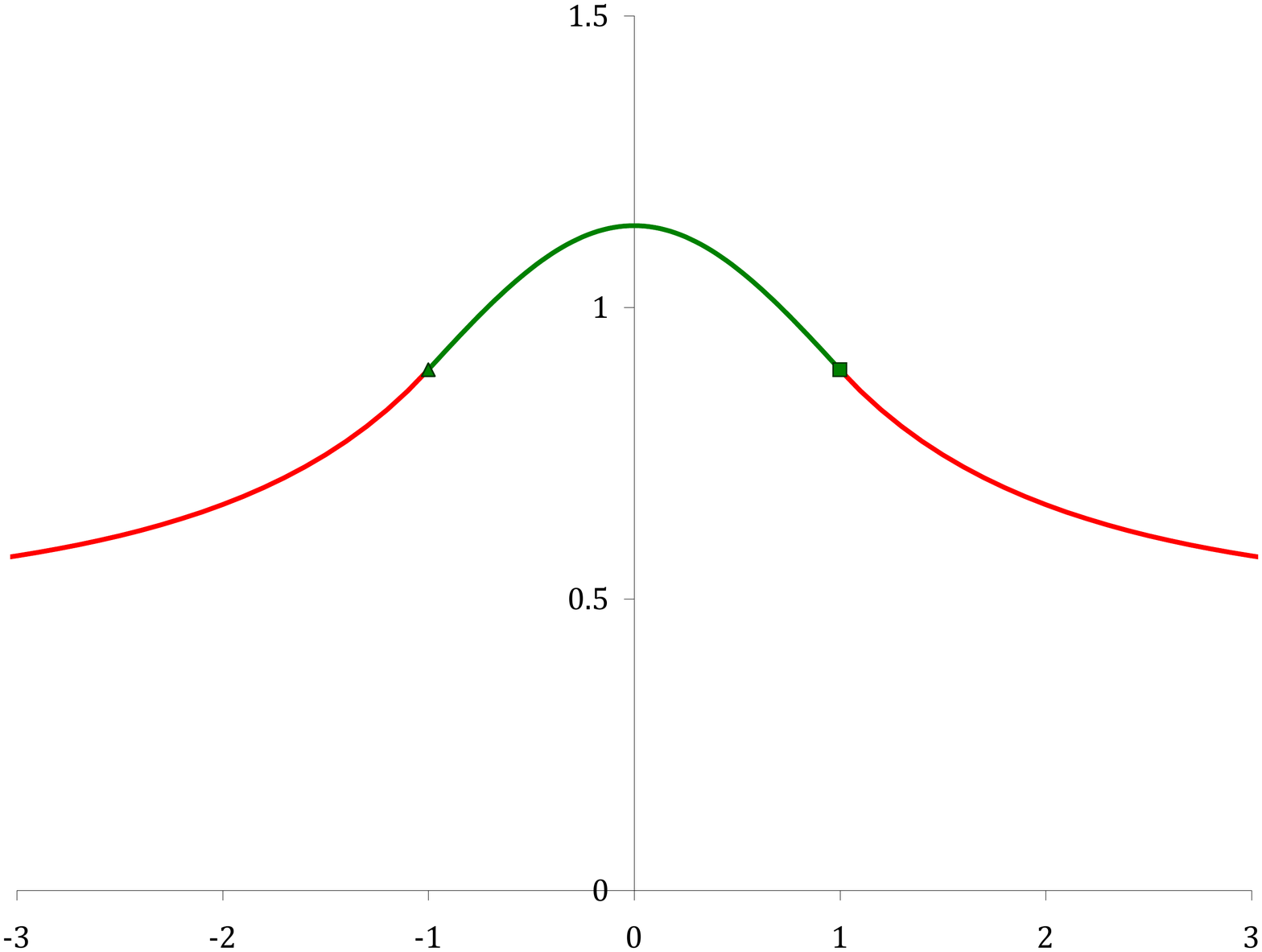}}
\subfigure{\includegraphics[width=2.3in]{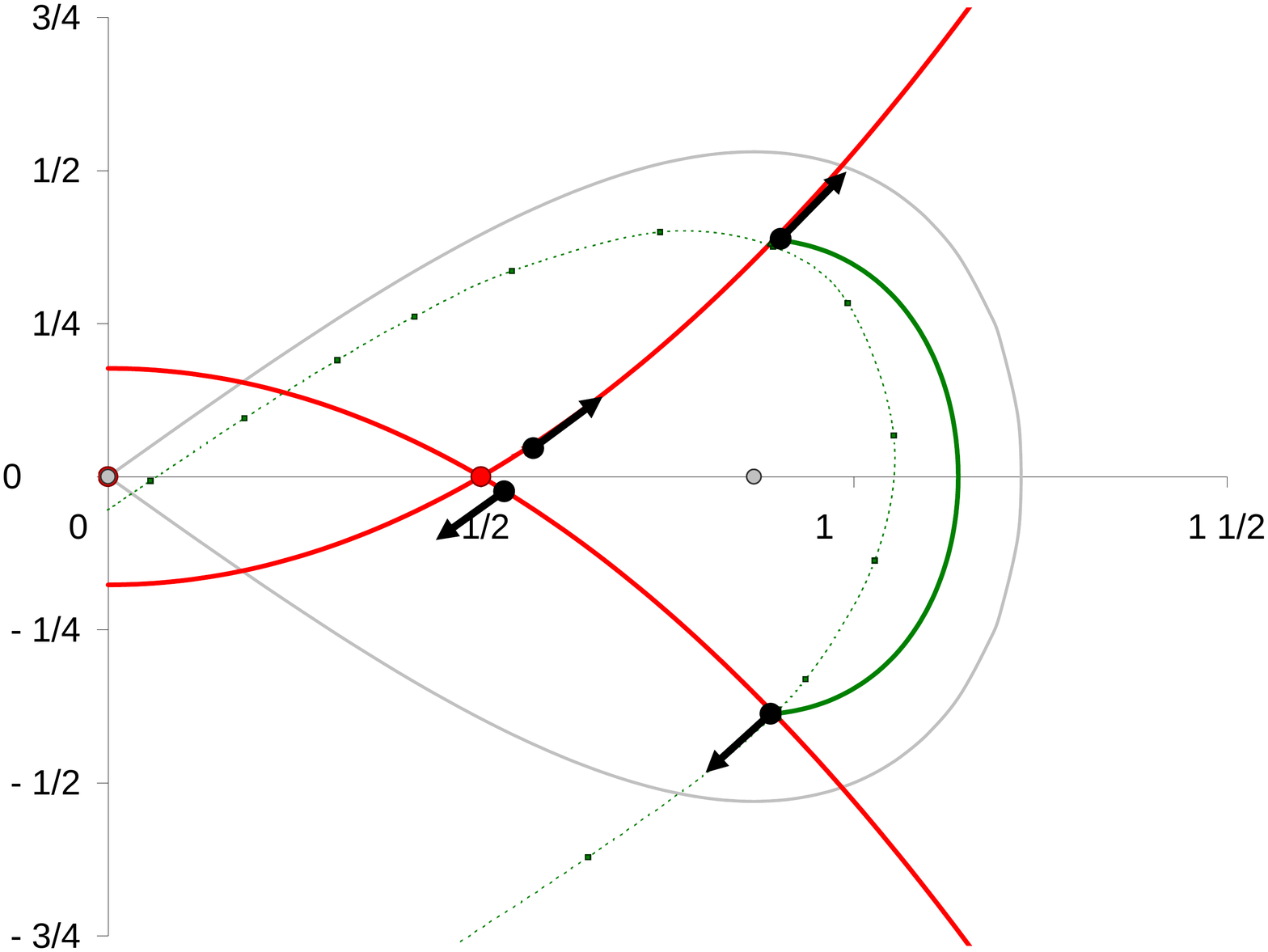}}
\end{center}
\caption{For $V=1$, $L=1$ and $r^2=0.25$, there are %In the $uv$-phase plane we find
two symmetric solutions, a ``dip'' and a ``bubble''.  Their profiles ($u$ versus $x$) are shown in the left panels, and their corresponding phase portraits and a cartoon demonstrating the computation of $P$ for each pulse are shown in the right panels.  For the dip, $P=0$ (no rotation), and for the bubble $P=1$ (rotation by $\pi$).  Since $Q=0$ for positive pulses, we know immediately that the bubble is unstable with at least one real positive eigenvalue.}
\label{fig:winding_small}
\end{figure}

For fixed parameters $r^2$, $V$ and $L$, consider the curve $\Gamma$ obtained by evolving the set of points on the (upward) parabola
$$\left\{(u_0,(u_0^2-r^2)/\sqrt{2}): u_0 \geq \sqrt{(V-\sqrt{V^2-2r^4})/2}\right\}$$
for time $2L$ via the inner system (\ref{stat2b}).  Intersections of this curve with the downward parabola $(u_0,-(u_0^2-r^2)/\sqrt{2})$ represent standing waves.  Suppose that the ``dip'' and the ``bubble'' are the only two standing waves present.  The end of $\Gamma$ near the cut-off must lie below the downward parabola, since the endpoint of $\Gamma$ is set up to be on the stable manifold of $(0,0)$ in the inner system.  Since the ``dip'' is assumed to exist for these parameter values, it must correspond to the intersection of $\Gamma$ and the downward parabola nearest this end, where $\Gamma$ passes from under to over the downward parabola.  Since $\Gamma$ is a smooth curve, if there are only two intersections, the other, corresponding to the ``bubble'', must occur as $\Gamma$ passes back from over to under the downward parabola.  Because vectors that are tangent to the upward parabola for $x=-L$ will evolve to be tangent to $\Gamma$ at $x=+L$, we can immediately conclude from these crossings that $P=0$ for the dip and $P=1$ for the bubble.  In particular, the bubble is unstable with at least one positive real eigenvalue. See Figure \ref{fig:winding_small}.

\begin{figure}[tb!]
\begin{center}
%\subfigure{\includegraphics[width=4.0in]{robby_diagrams_b}}
%\subfigure{\includegraphics[width=2in]{robby_diagrams_b1p}}
%\subfigure{\includegraphics[width=2in]{robby_diagrams_b1}}
\subfigure{\includegraphics[width=2.3in]{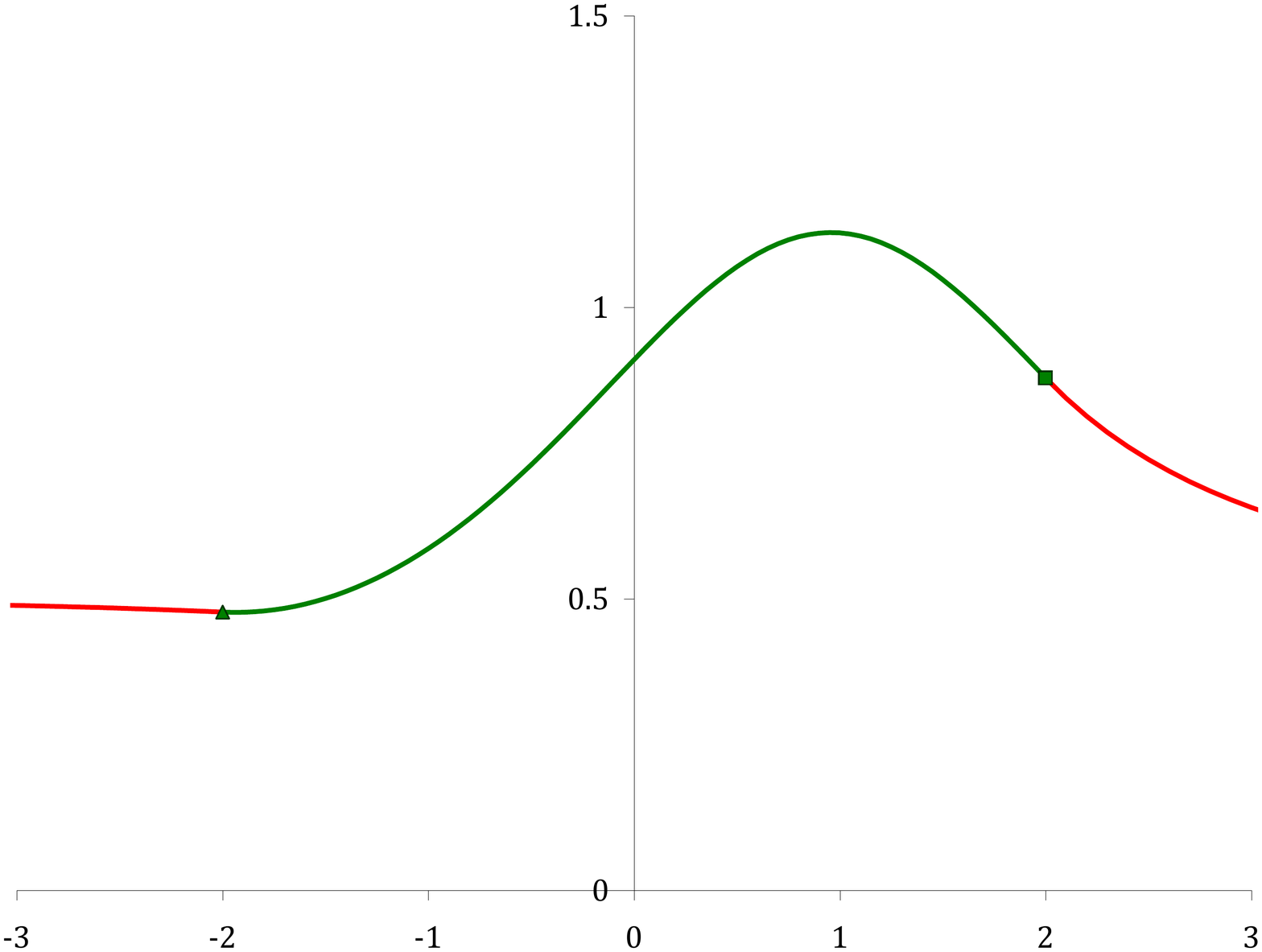}}
\subfigure{\includegraphics[width=2.3in]{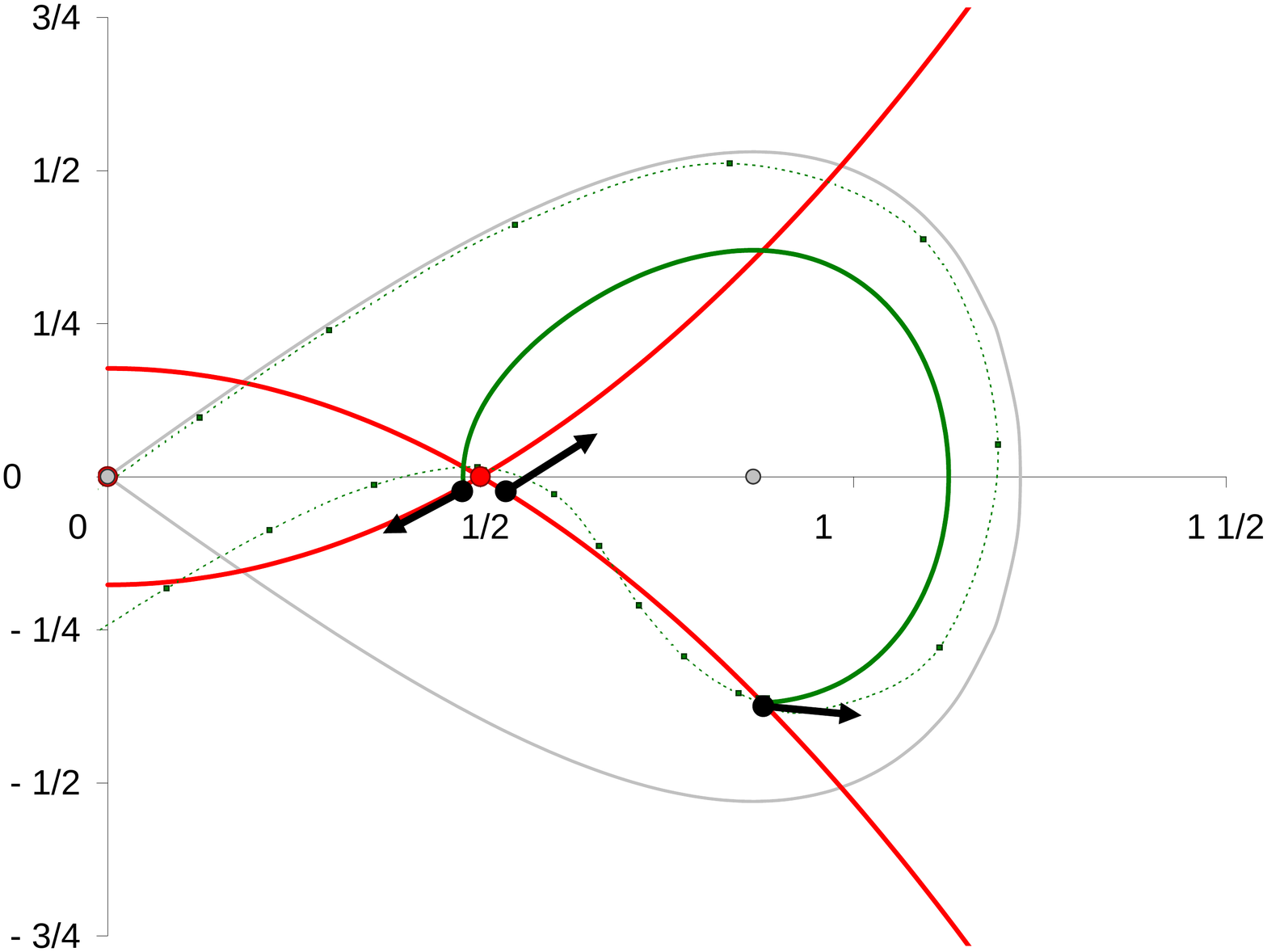}}
\subfigure{\includegraphics[width=2.3in]{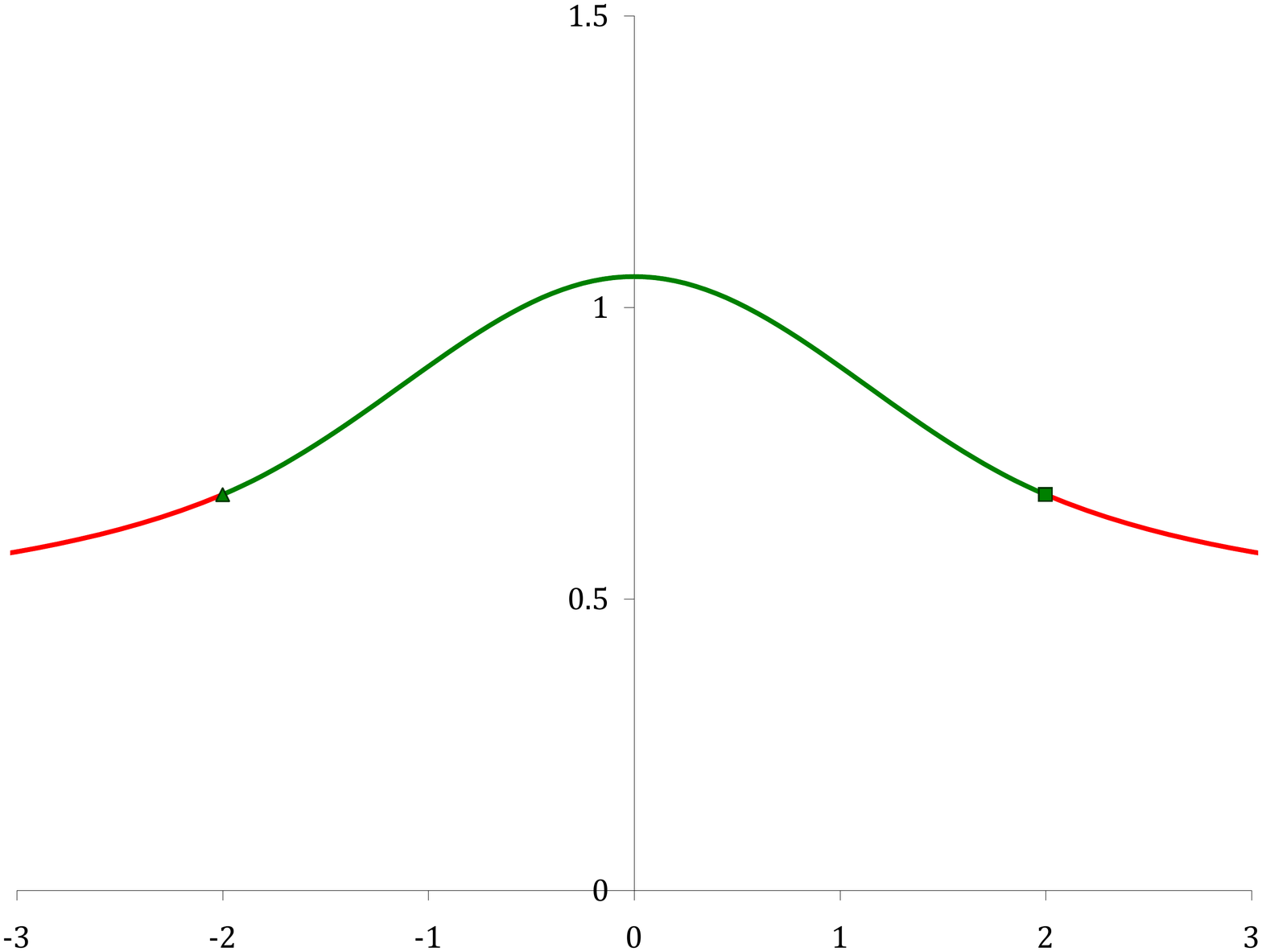}}
\subfigure{\includegraphics[width=2.3in]{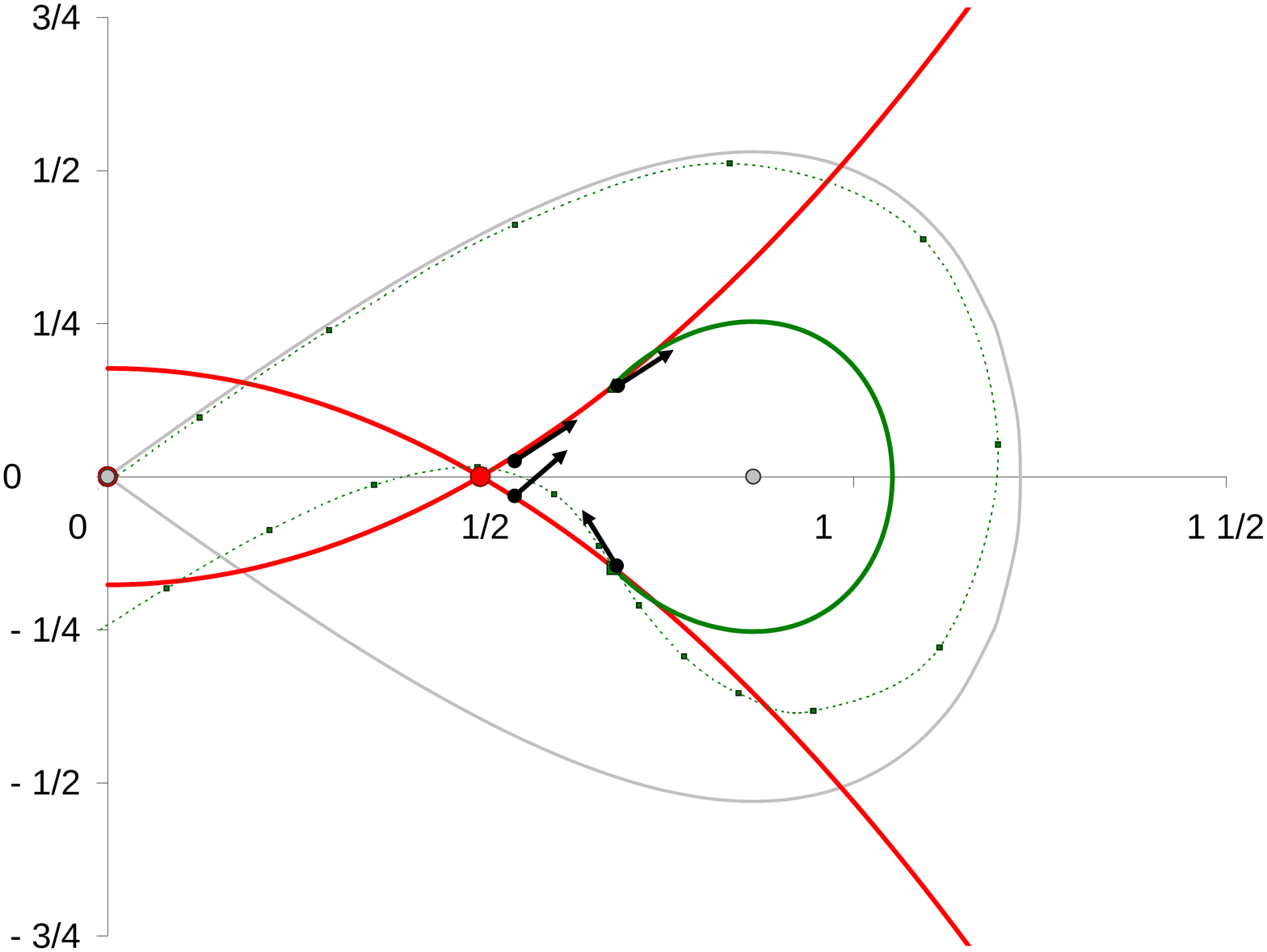}}
\end{center}
\caption{The same as in Figure \ref{fig:winding_small}, but for $V=1$, $L=2$ and $r^2=0.25$. Here, there are %In the $uv$-phase plane we find
two symmetric solutions and a pair of asymmetric solutions.  (Because the two asymmetric solutions overlap significantly in the phase plane, we show only one.)  Plots and computations of the symmetric ``bubble'' and one of the asymmetric pulses are shown. %at the left, and a cartoon demonstrating the computation of $P$ for each pulse is shown on the right.
For the asymmetric pulse $P=1$ (rotation by $\pi$)and for the bubble $P=2$ (rotation by $2\pi$).  Since $Q=0$ for positive pulses, we know immediately that the asymmetric pulse is unstable with at least one real positive eigenvalue and the bubble is unstable with at least two real positive eigenvalues.}
\label{fig:winding_large}
\end{figure}

Now consider parameters $r^2$, $V$ and $L$, so that a ``dip'', a ``bubble'' and two asymmetric waves are present.  In this case, the transient ($-L<x<L$) portion of the ``bubble'' immediately moves above the upward parabola before passing back through the upward and downward parabolas and approaching the downward parabola from below as $x \rightarrow L$.  In particular, a vector $p(x)$ initialized to be tangent to the upward parabola along this solution for $x<-L$ will be a slight ``clockwise'' rotation away from a vector initialized tangent to the transient loop at $x=L$.  And, as the tangent vector evolves for $-L<x<L$, that vector remains tangent to the transient and $p(x)$ must continue to advance clockwise ahead of the tangent vector throughout this interval.  (Because solutions are unique, $p(x)$ cannot pass through the tangent vector is effectively trapped.)  When $x=L$, the tangent vector has rotated around to point into the region above the downward parabola.  Since $p(x)$ is a further clockwise rotation from there, this guarantees that $P \geq 2$ in this case.  However, adjacent intersections of $\Gamma$ with the upward parabola can only correspond to changes of $P$ by a unit amount, so since $P=0$ for the dip, we must have $P=2$ for the bubble here (and $P=1$ for the intermediate intersection, representing one of the asymmetric standing waves).  See Figure \ref{fig:winding_large}. \qed
\end{proof}

In Figure \ref{fig:bif} we describe and summarize Theorems \ref{th:exist} and \ref{th:stability} for the positive solutions in Figs.\ \ref{fig:winding_small} and \ref{fig:winding_large}. When $L=1$ only the two positive-valued symmetric pulses are present. For large enough $L$ (depending upon $V$ and $r^2$), there are asymmetric standing waves. When $L=2$, we see the asymmetric branch appearing. The location of the pitchfork bifurcation in the second panel has been discussed before (see Figure \ref{fig:L}). The numbers besides the curves show the value of $P-Q$ for the corresponding solutions.

\begin{figure}[tbhp!]
\begin{center}
\subfigure{\includegraphics[width=3.0in]{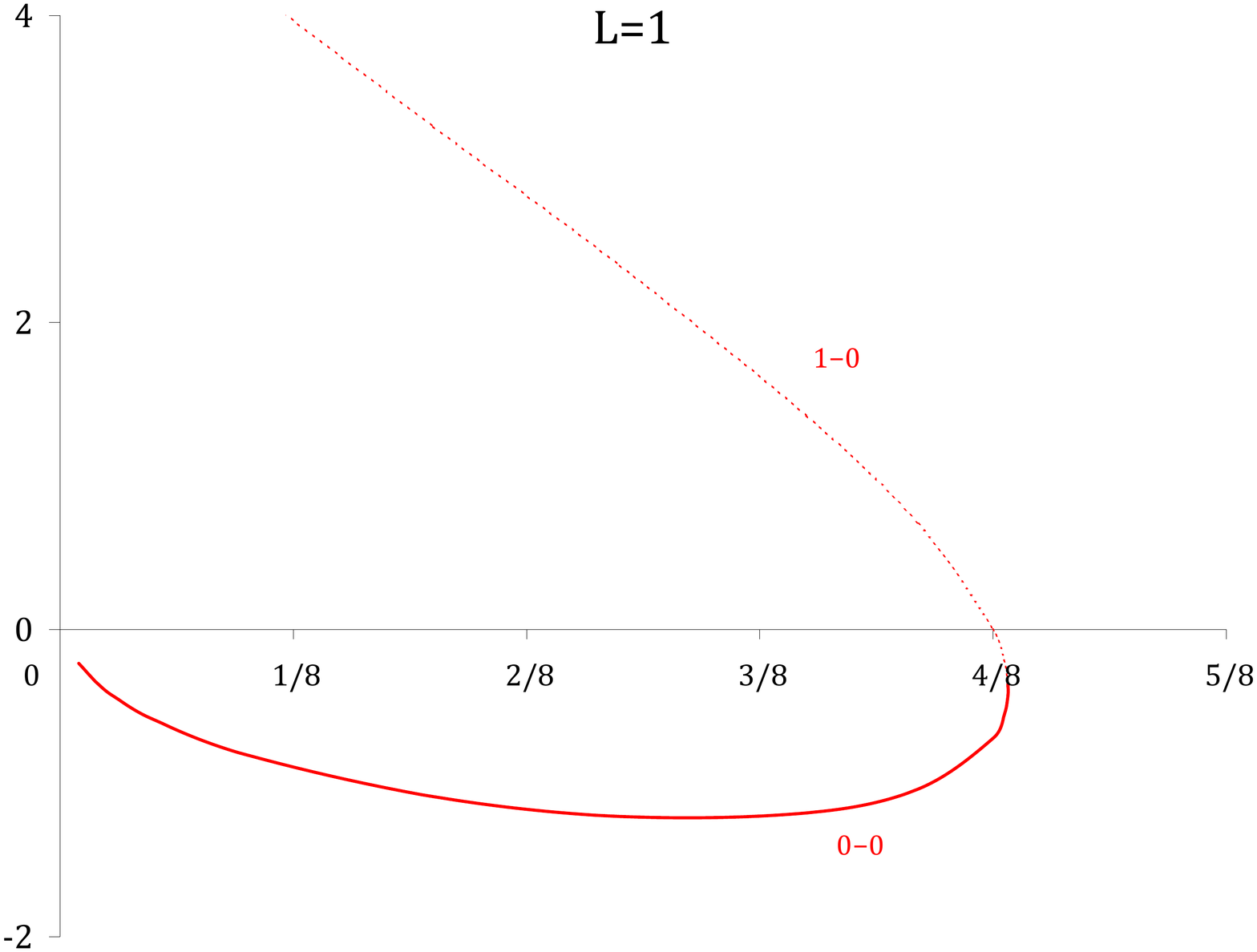}}
\subfigure{\includegraphics[width=3.0in]{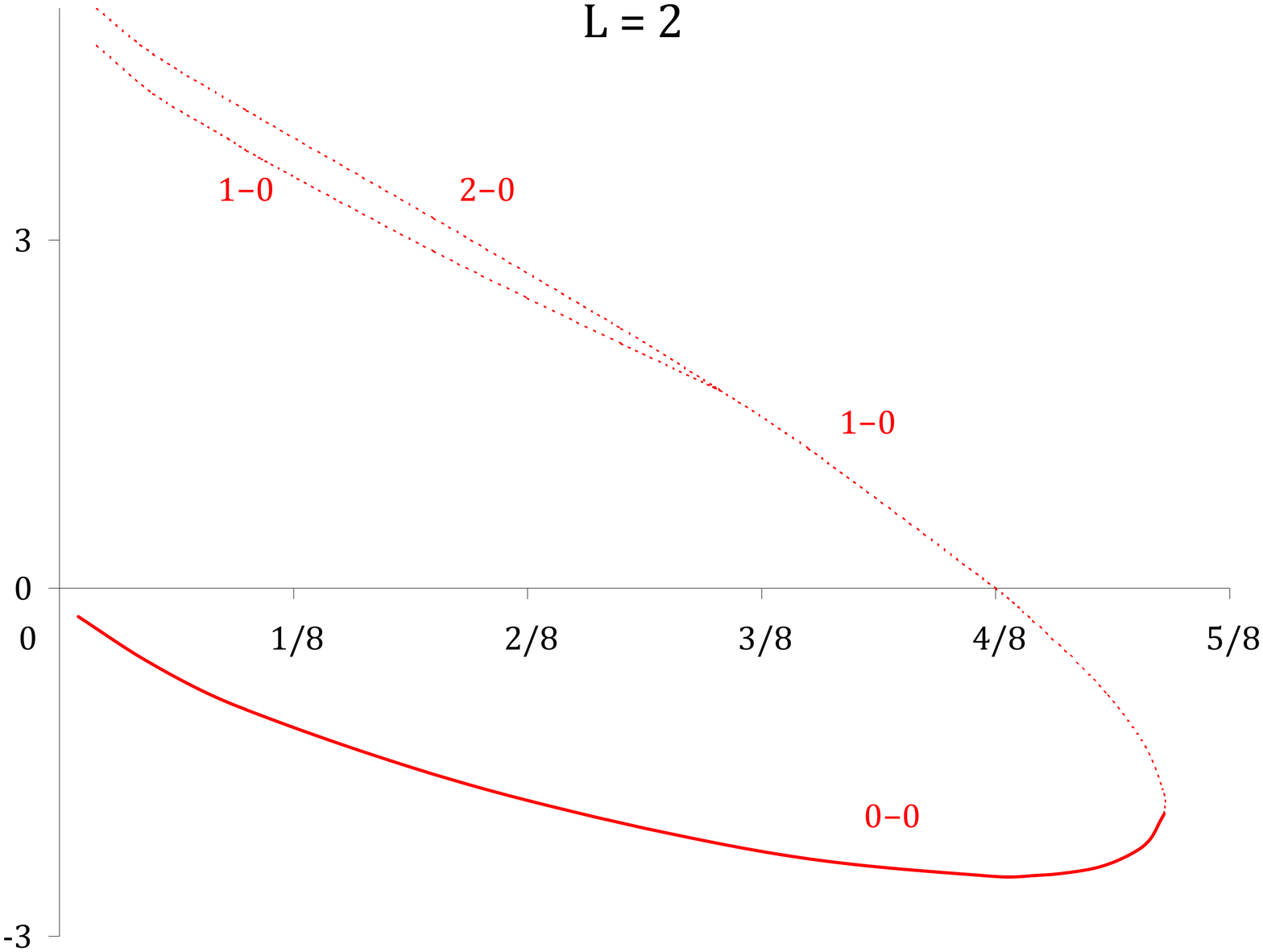}}
\end{center}
\caption{Bifurcation and stability diagrams ($N$ versus $r^2$) for positive standing waves in Figs.\ \ref{fig:winding_small} and \ref{fig:winding_large}. }
\label{fig:bif}
\end{figure}

\section{Numerical results: positive solutions}
\label{sec:1}

\begin{figure}[tb!]
\begin{center}
\subfigure[the ``dip'' of Fig.\ \ref{fig:winding_small}]{\includegraphics[scale=.4]{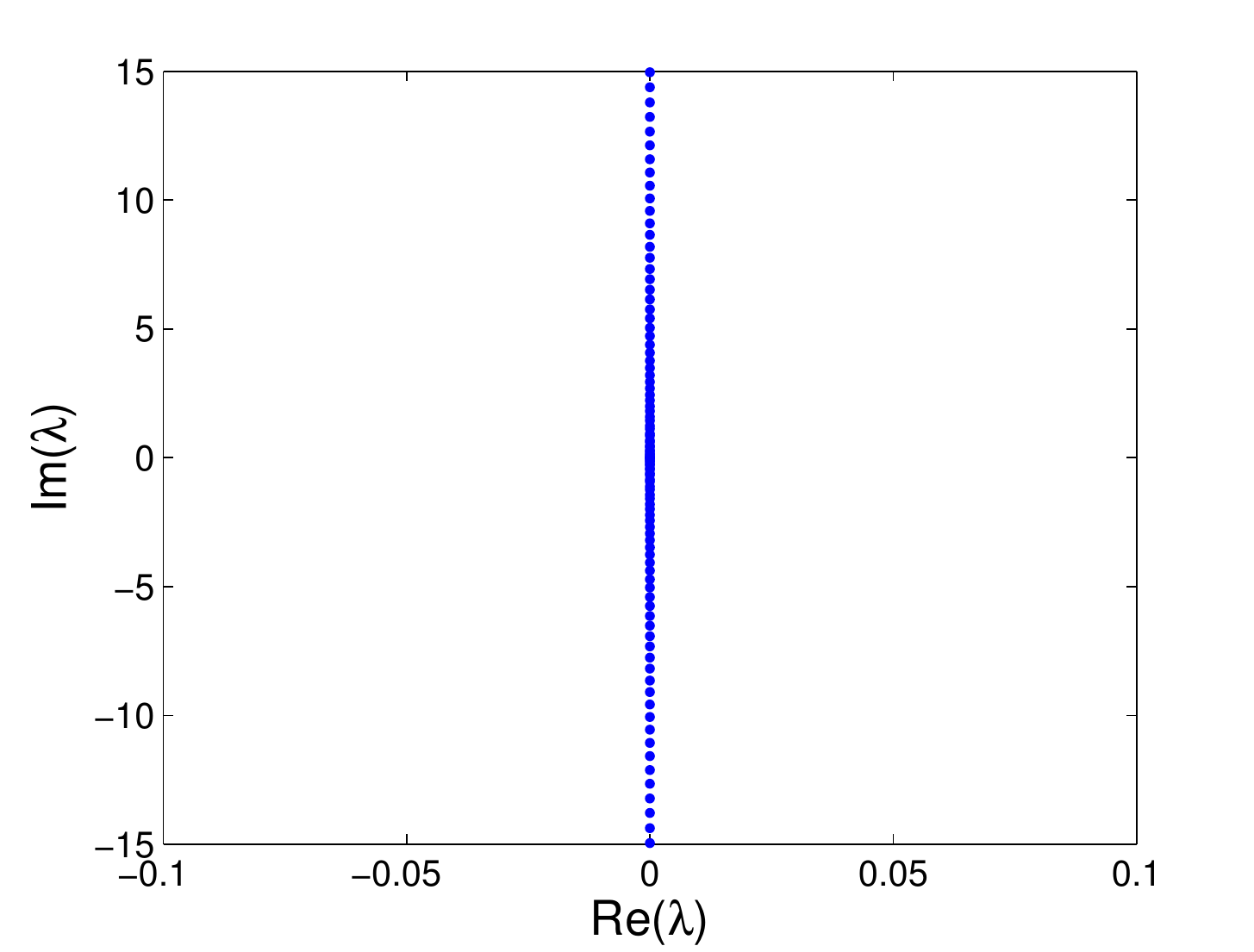}}
\subfigure[the ``bubble'' of Fig.\ \ref{fig:winding_small}]{\includegraphics[scale=.4]{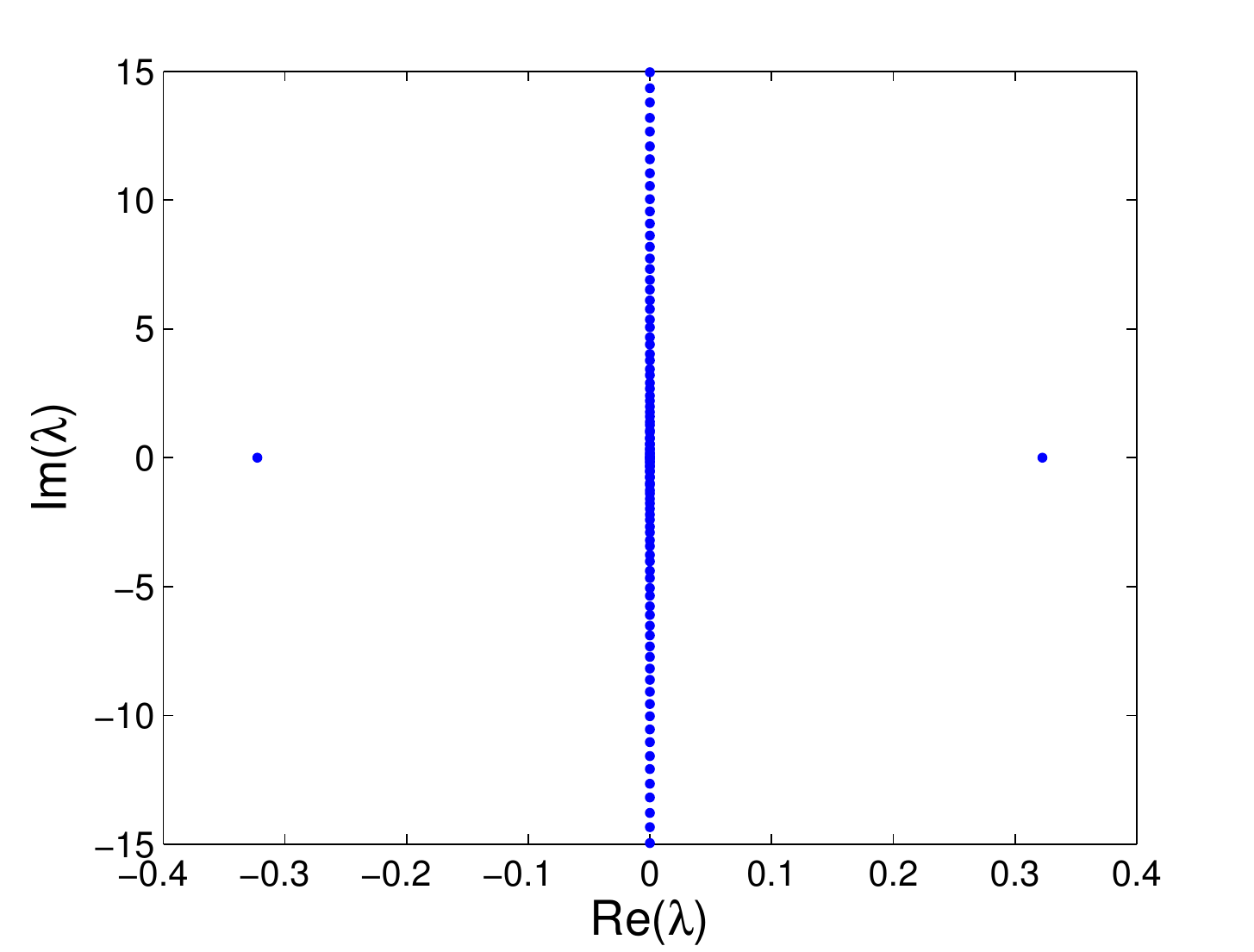}}
\subfigure[the asymmetric ``bubble'' of Fig.\ \ref{fig:winding_large}]{\includegraphics[scale=.4]{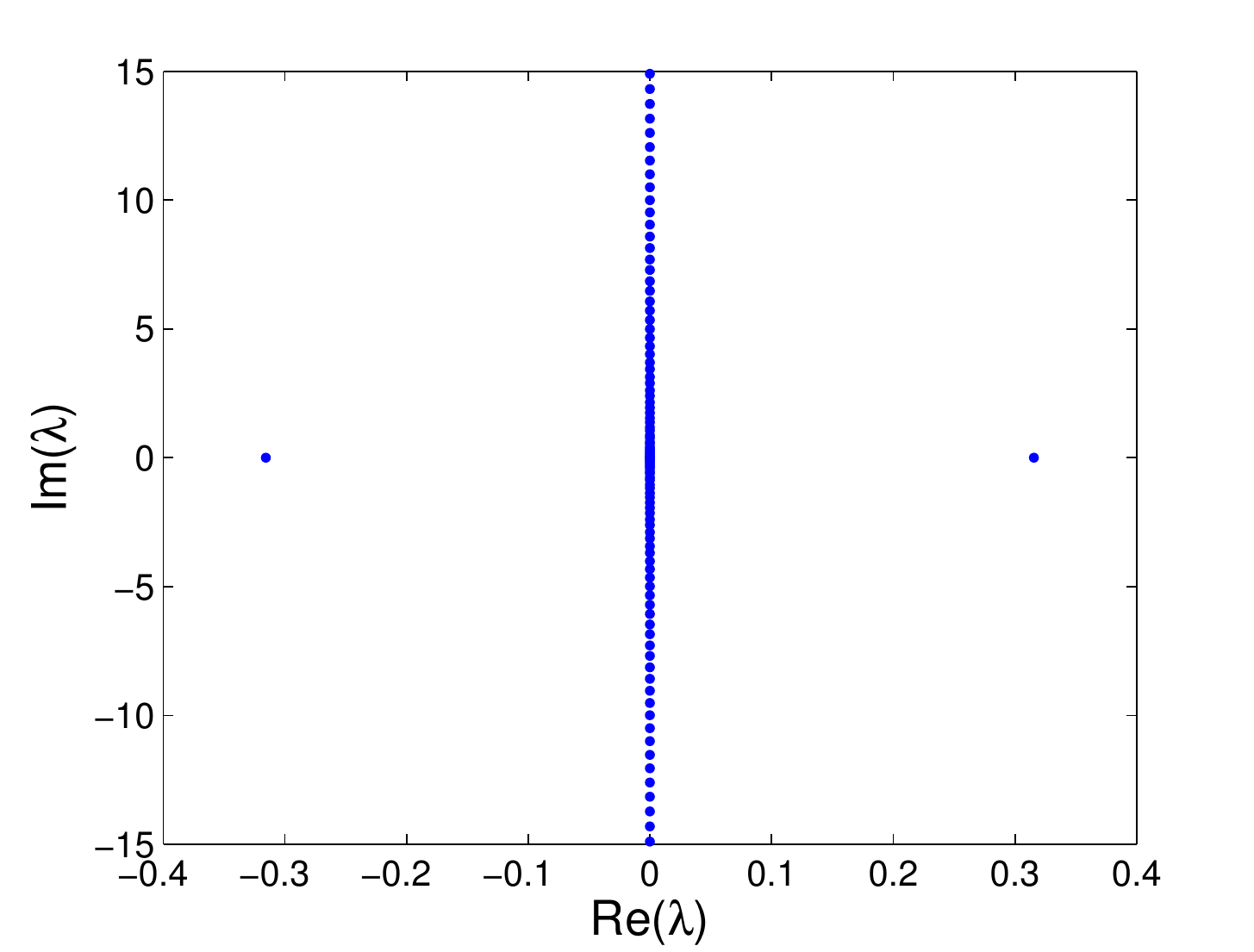}}
\subfigure[the symmetric ``bubble'' of Fig.\ \ref{fig:winding_large}]{\includegraphics[scale=.4]{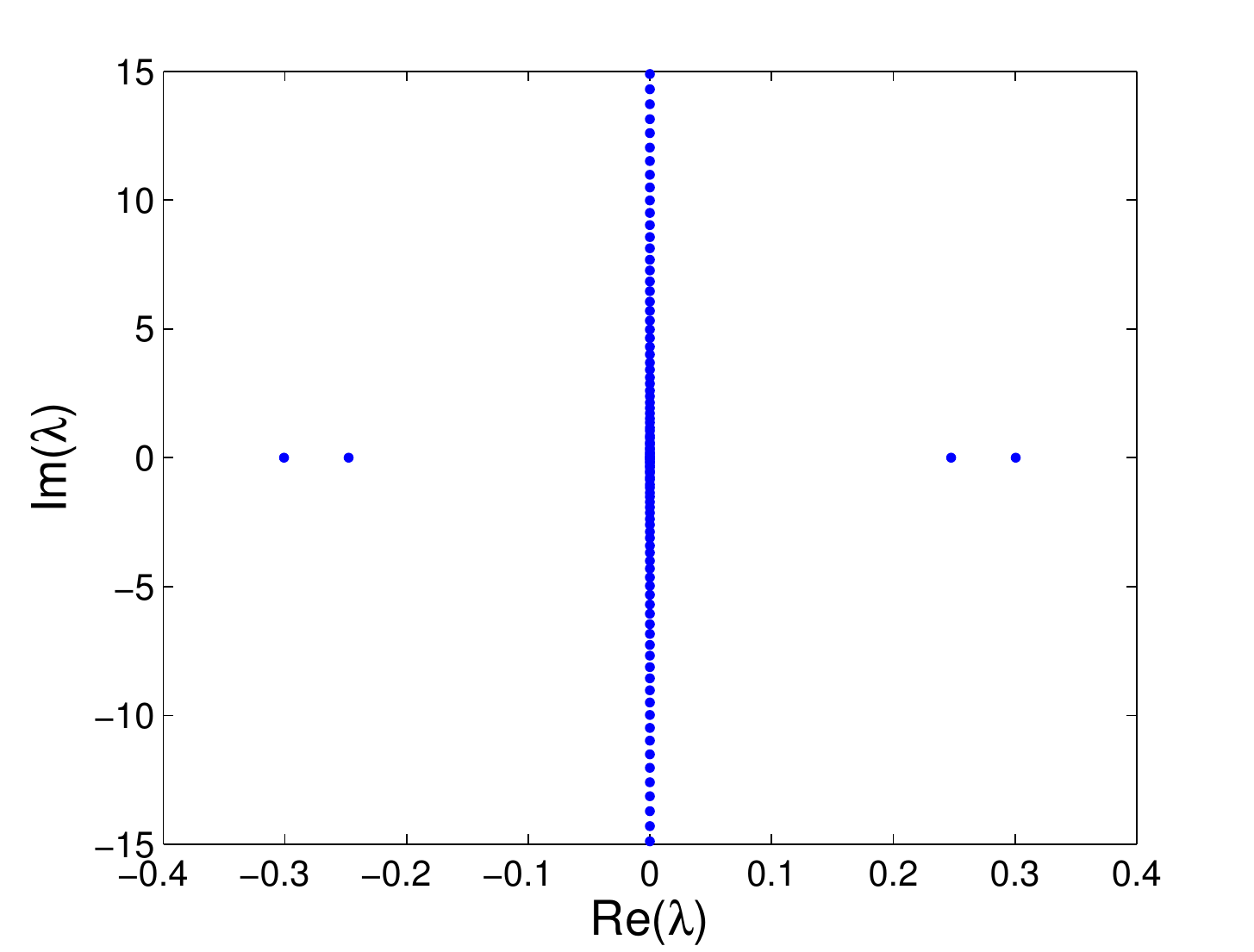}}
\end{center}
\caption{The eigenvalue structures of the solutions in Figures \ref{fig:winding_small} and \ref{fig:winding_large} as indicated in the caption of each panel. %When $r^2=0.55$, there are two positive-valued symmetric standing waves on a nonzero background. The pulse profiles are shown at the top; note that (with $r^2>0.5$), the ``dip'' and ``bubble'' solution both have the appearance of a dip.  The spectrum distributions in the complex plane of each pulse is also shown.  As required by Theorem \ref{th:stability}, the standing wave labeled as a bubble has a real positive eigenvalue.
}
\label{fig:sol1}
\end{figure}

In order to corroborate the analytical %instability
results of the previous section, we numerically solve the time-independent equation (\ref{stat1}) using a Newton-Raphson method and the eigenvalue problem (\ref{eq:linear}) for a selection of standing waves. Additionally, to study the typical dynamics of an unstable solution, the time-dependent equation (\ref{gov1}) is then integrated numerically using a fourth-order Runge-Kutta method with a numerically exact unstable solution obtained from the time-independent equation as the initial condition.

\begin{figure}[tbhp!]
\begin{center}
{\includegraphics[scale=.4]{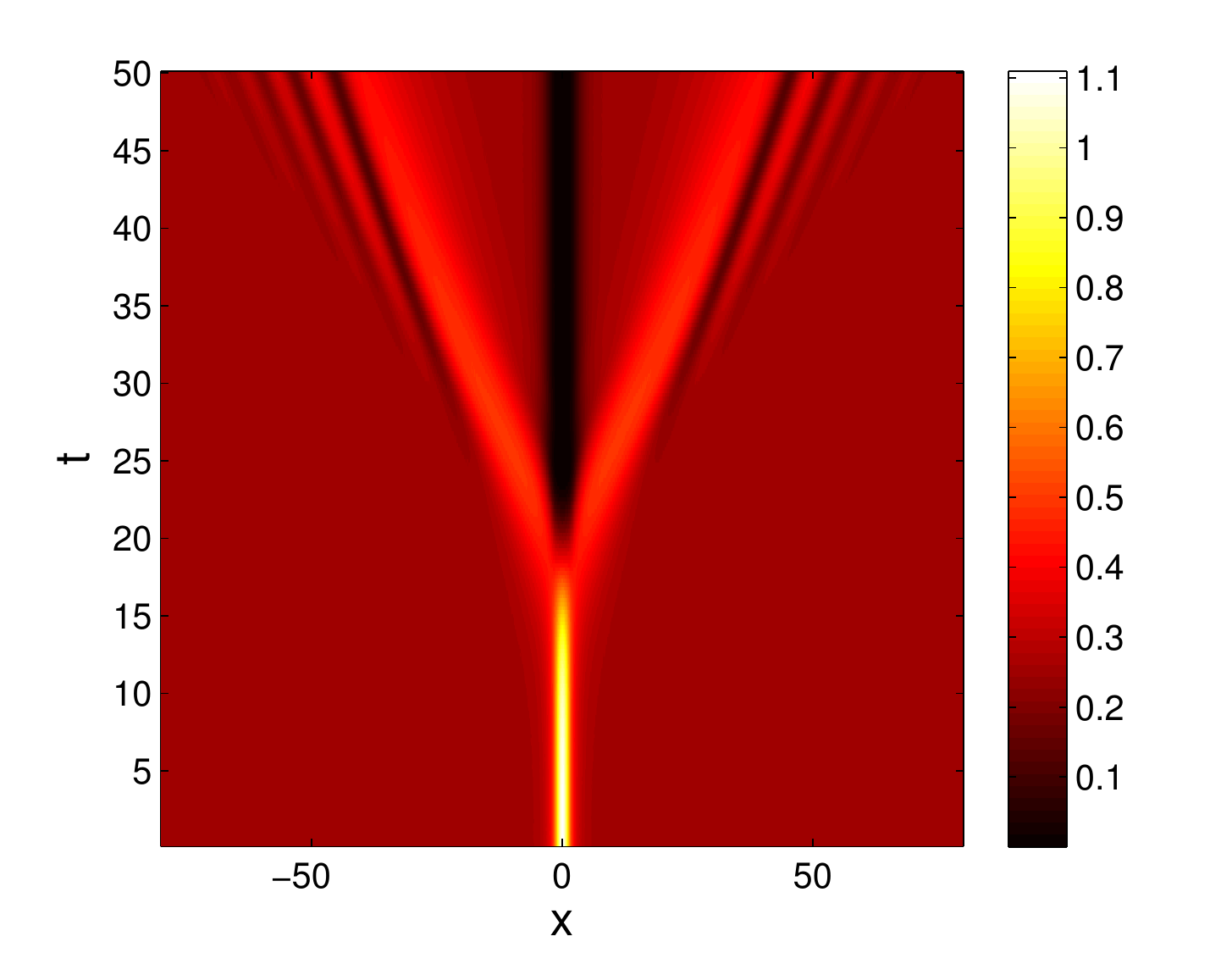}}
{\includegraphics[scale=.4]{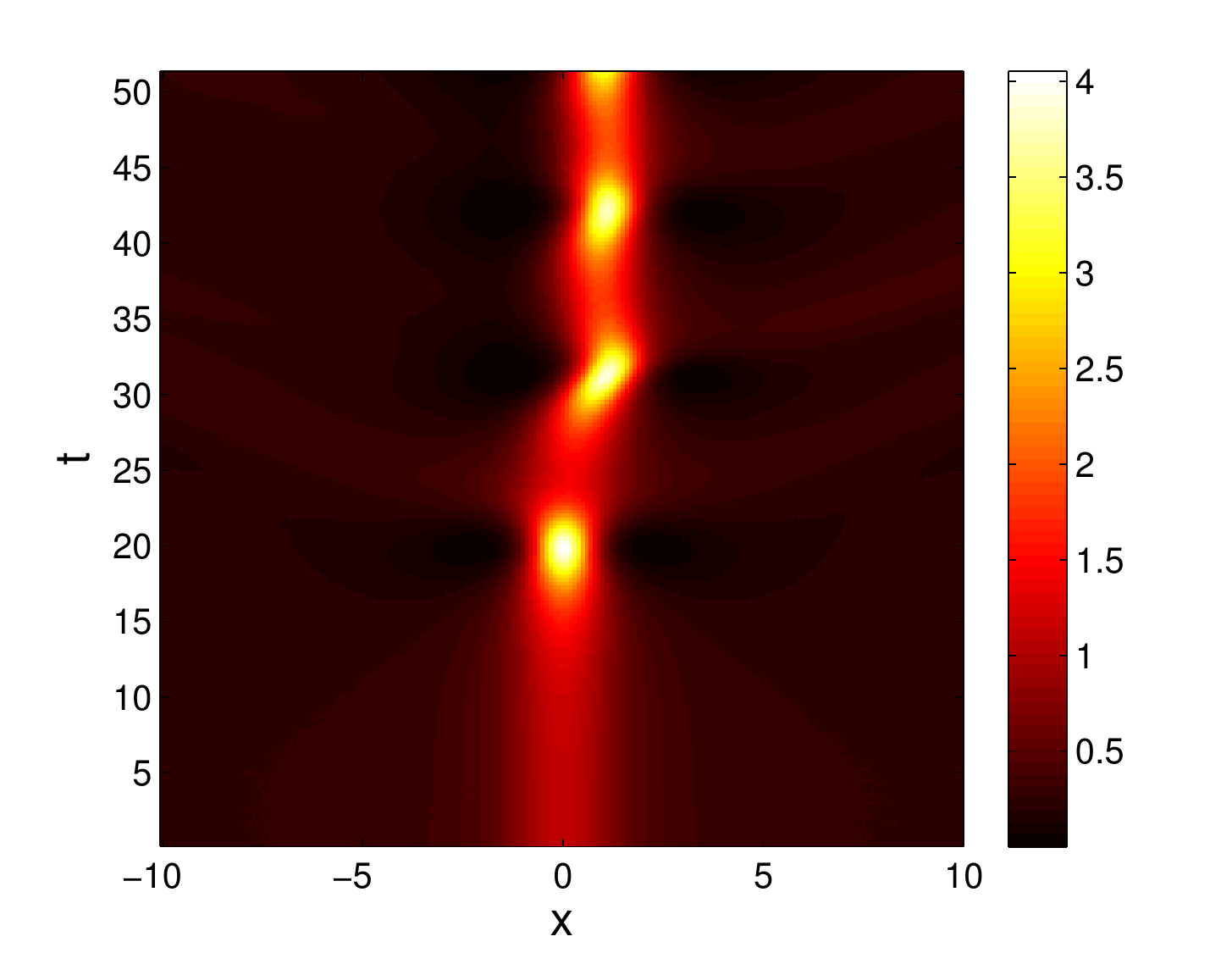}}
\end{center}
\caption{Typical time evolutions of ``bubble'' solutions when $V=1$, $L=2$, and $r^2=0.25$ (two positive real eigenvalues). Shown is the top-view of $|\Psi(x,t)|^2$.}
\label{fig:evol1}
\end{figure}

Shown in Figure \ref{fig:sol1} are the numerically computed eigenvalues of the solutions in Figures \ref{fig:winding_small} and \ref{fig:winding_large} where
%In each case, you can see in Figures \ref{fig:sol1} and \ref{fig:sol2} that
the eigenvalue counts come in in full agreement with Theorem \ref{th:stability}. As predicted in the theorem, when the asymmetric standing wave is present, it is unstable with at least one real positive eigenvalue and the symmetric ``bubble'' picks up a second real positive eigenvalue. In fact, ``bubble'' solutions have exactly $|P-Q|$ real positive eigenvalues present. The ``dip'' solution, which is the only positive-valued solution not declared unstable \emph{a priori} because $P-Q=0$, is in fact stable (see Remark \ref{rem1}).

Two typical dynamics of an unstable symmetric ``bubble'' solution (in this case the symmetric solution in Fig.\ \ref{fig:winding_large}) are shown in Fig.\ \ref{fig:evol1} depending on the type of perturbations introduced in the initial conditions of the time integration of Eq.\ (\ref{gov1}). With the initial condition $\Psi(x,0)=\left(1\pm10^{-3}\right)u(x)$ and $u(x)$ being a numerically exact solution, when the sign is minus, one would obtain that the ``bubble'' evolves into a ``dip'' and releases two traveling humps that later develop shock waves. This is depicted in the left panel of the figure. When the sign is plus, as a manifestation of the instability, a solitary wave is clearly created as shown in the right panel of the figure, which is akin to the formation of bright solitons due to modulational instability \cite{abdu05,kevr04}. %the ``bubble'' becomes an asymmetric pulse with an oscillating amplitude.

\section{Excited solutions}

\begin{figure}[tbhp!]
\begin{center}
\subfigure[$r^2=0.80$]{\includegraphics[scale=.4]{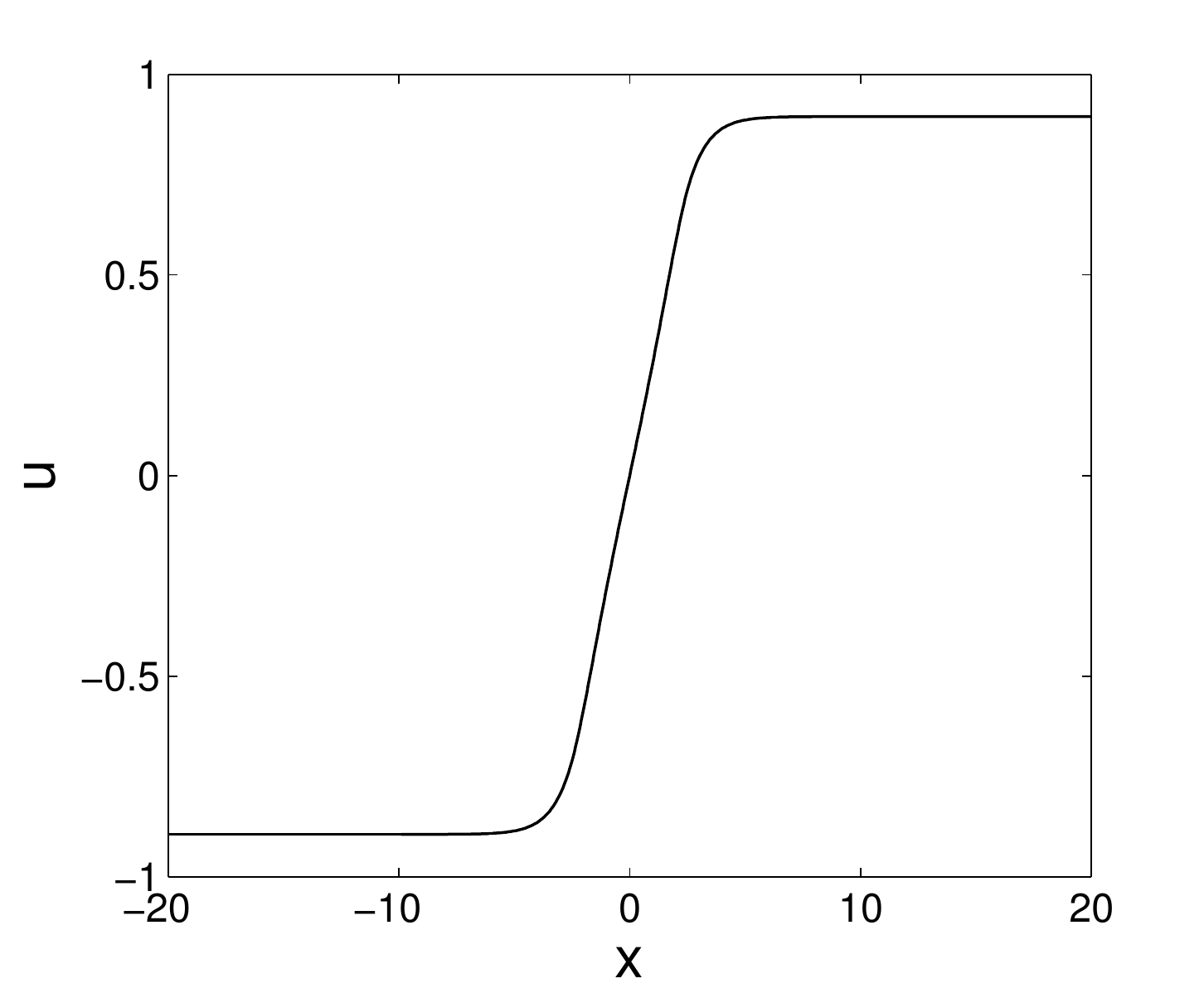}}
\subfigure[$r^2=0.86$]{\includegraphics[scale=.4]{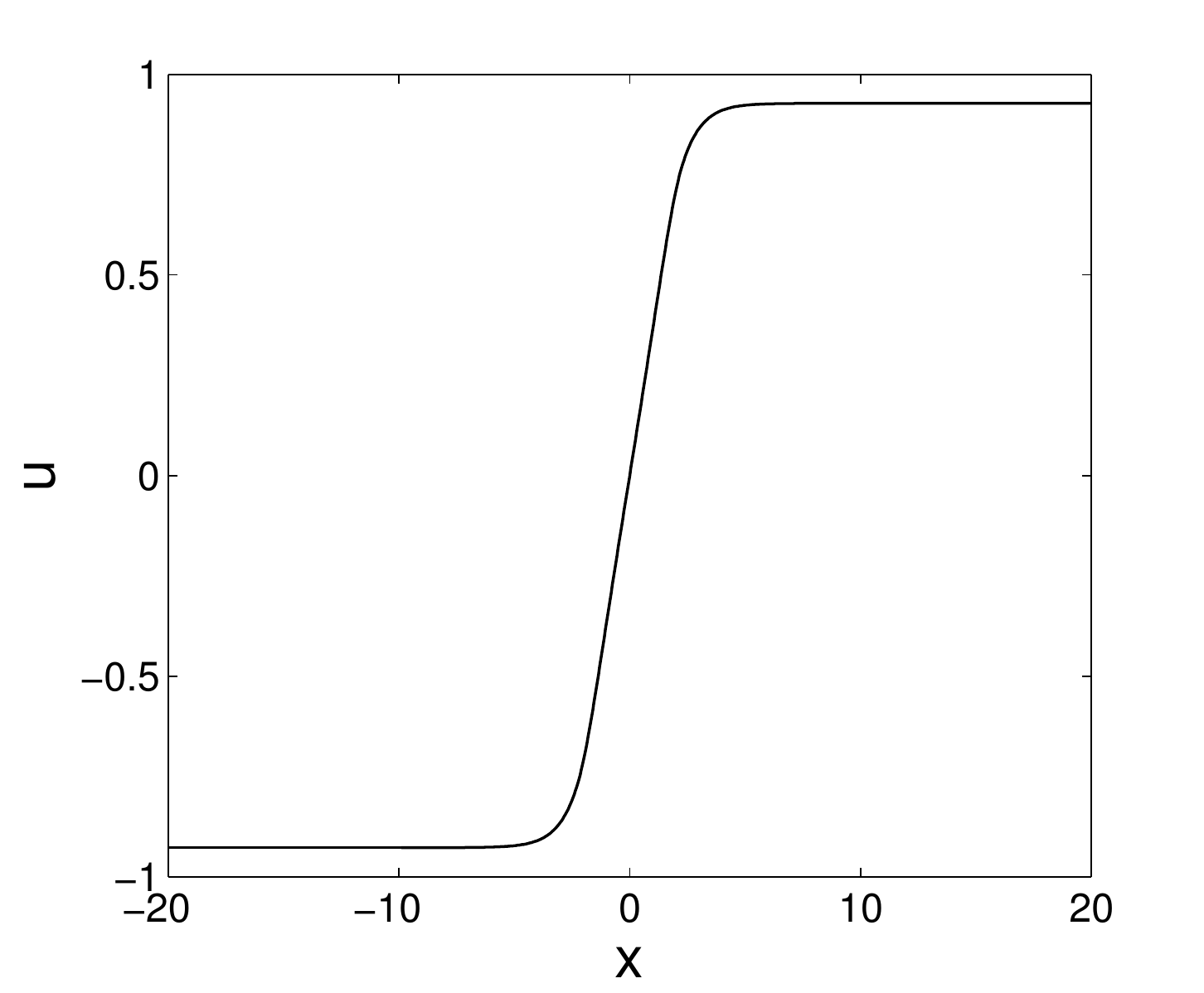}}
\subfigure[$r^2=0.80$]{\includegraphics[scale=.4]{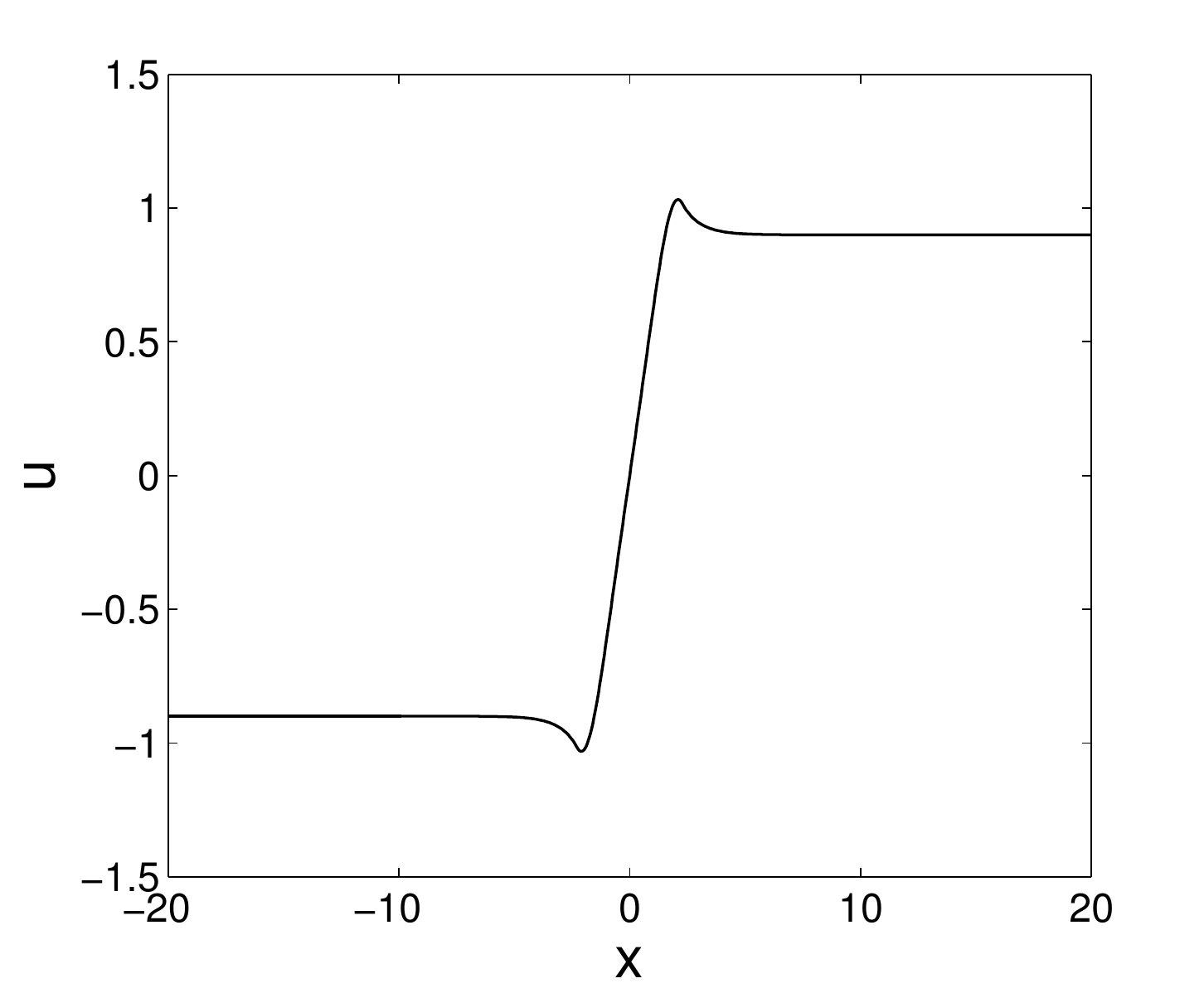}}
\subfigure[$r^2=0.60$]{\includegraphics[scale=.4]{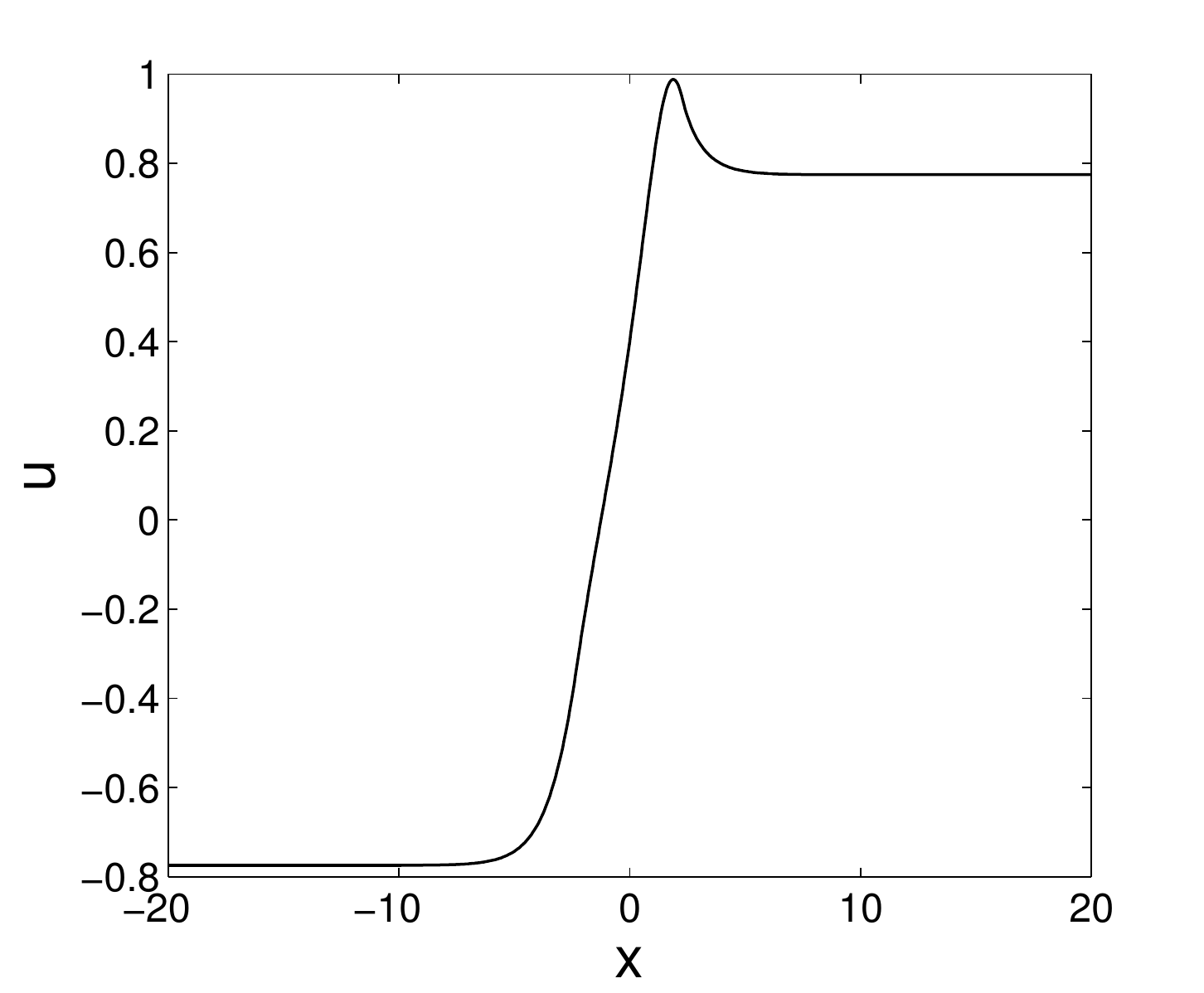}}
\subfigure[$N$ vs.\ $r^2$]{\includegraphics[scale=.6]{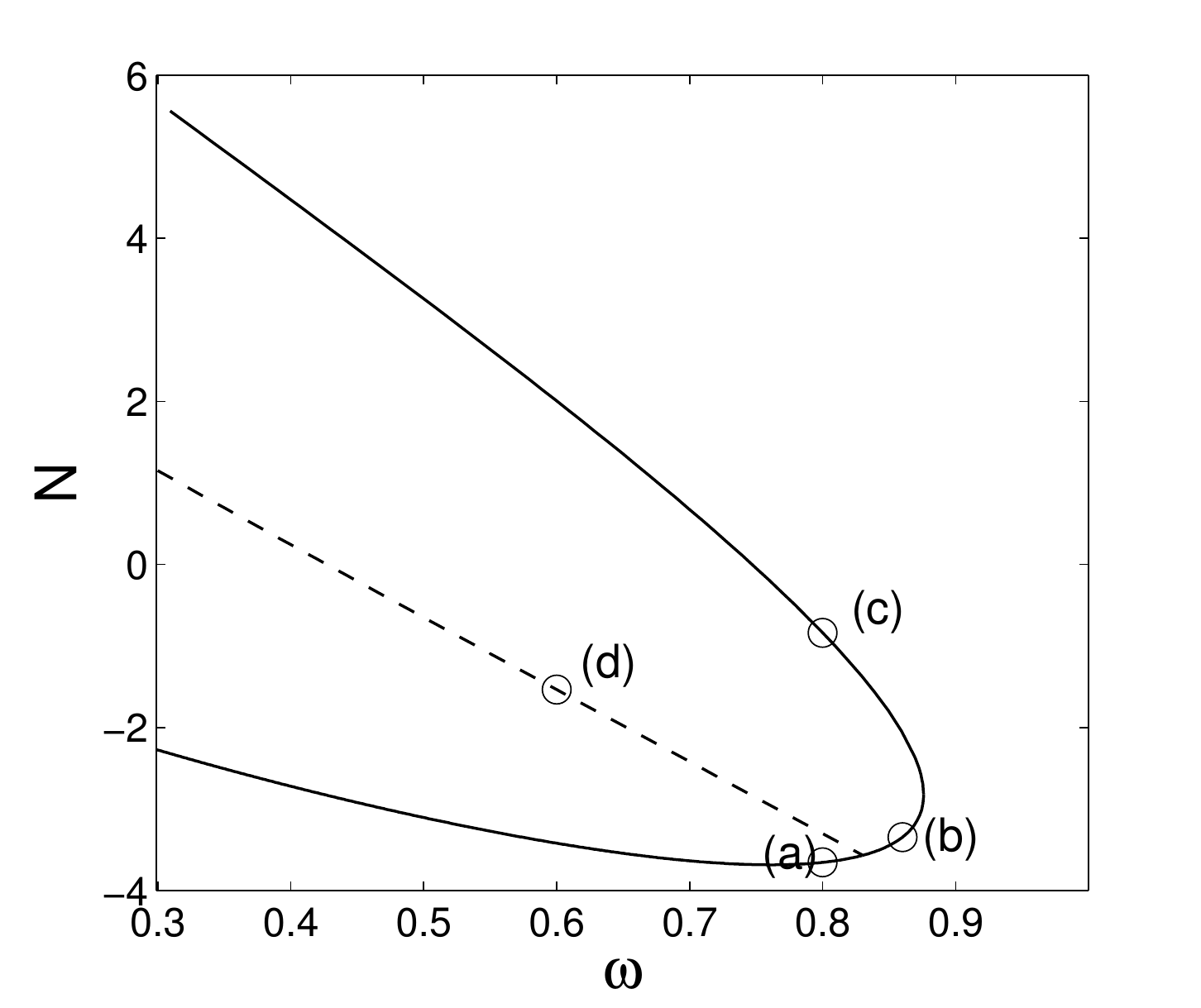}}
\end{center}
\caption{(a-d) Four black solitons with the constant $r^2$ mentioned in each caption. In addition to the symmetric states, there is also an asymmetric state. (e) The norm $N$ as a function of $r^2$ for the first excited states. The computation of $P-Q$ is also indicated near the curves. }
\label{fig:truedark}
\end{figure}

In addition to positive solutions, %Although we did not discuss it above,
we also consider standing waves on a nonzero background that have a zero in the profile, \textit{i.e.}, black solitons.  Such solutions can be identified using the same method above.  In Fig.\ \ref{fig:truedark}, we present four examples of such solutions with the corresponding value of $r^2$ shown in the caption of each panel and their bifurcation diagram. As was the case with the positive solutions, we can identify two symmetric solutions, as well as a pair of asymmetric solutions. %, as depicted in Fig.\ \ref{fig:sol2}(d).
In this case, however, the asymmetric solutions bifurcate from the lower symmetric branch.

After obtaining the general picture of the existence of static solutions, next we study the stability of them. Note that for black solitons, $Q=1$. Counting the number of positive eigenvalues of $D_+$, we obtain that in Figure \ref{fig:truedark} along the lower solid branch, to the left of the bifurcation point of the asymmetric states, the symmetric black solitons have $P=0$. To the right of the pitchfork bifurcation, until the saddle-node, the operator $D_{+}$ has one positive eigenvalue, i.e.\ $P=1$. Turning up onto the upper branch, another positive eigenvalue of $D_{+}$ appears and $P=2$. As was the case with the positive states, the asymmetric state has $P=1$ in its existence region. Because $Q=1$, using Lemma \ref{th:main}, one can conclude that symmetric solutions along the upper branch and those along the lower branch to the left of the bifurcation point of asymmetric states are all unstable. As for the asymmetric state, as well as the symmetric ones on the part of the solid curve containing point (b) in Fig.\ \ref{fig:truedark}(e), {because $P-Q=0$ in these cases, and so the hypothesis of Lemma \ref{th:main} is not met,} we cannot conclude instability directly from the lemma. In this case, numerical calculations are necessary.

\begin{figure}[tbhp!]
\begin{center}
\subfigure[]{\includegraphics[scale=.4]{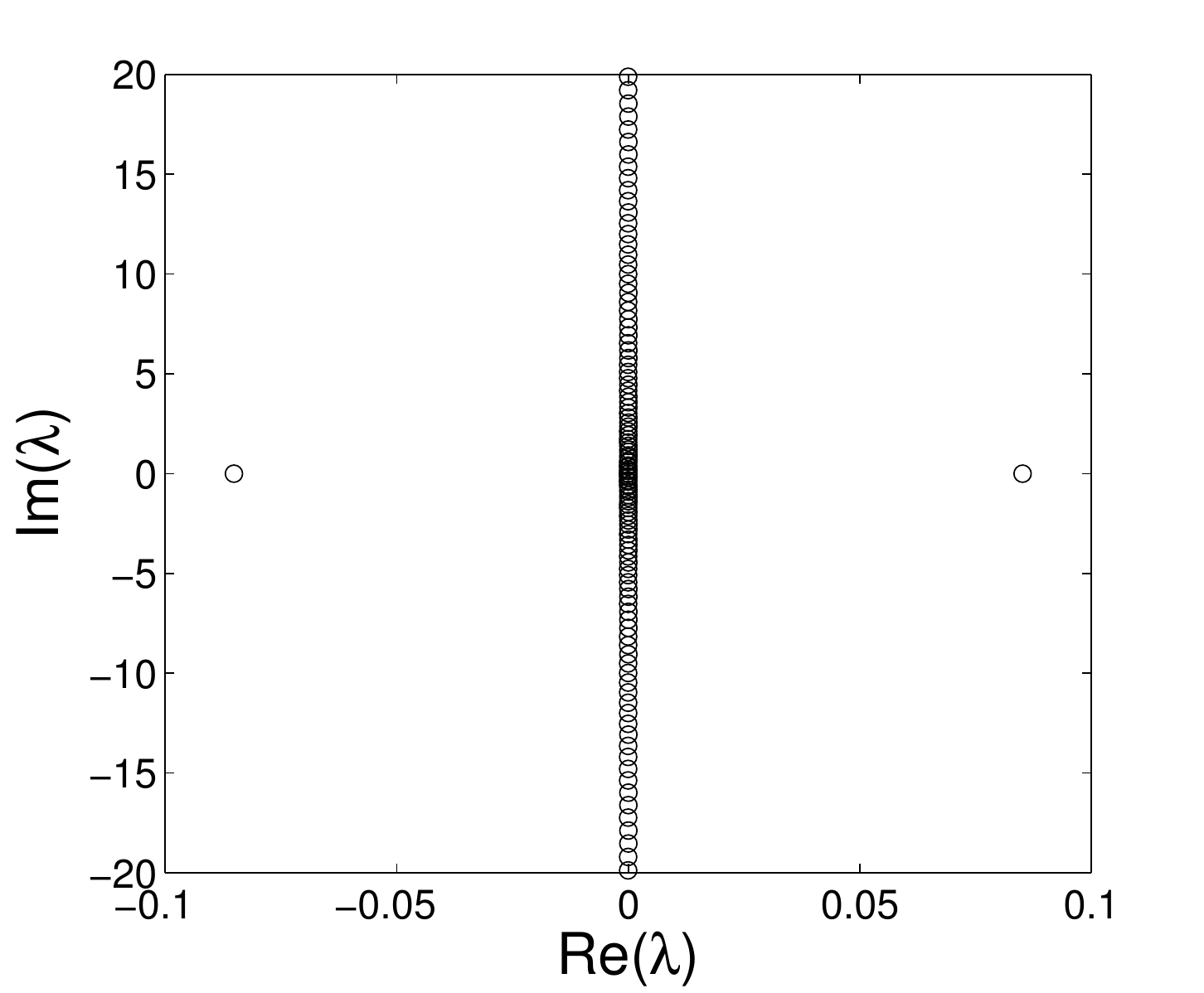}}
\subfigure[]{\includegraphics[scale=.4]{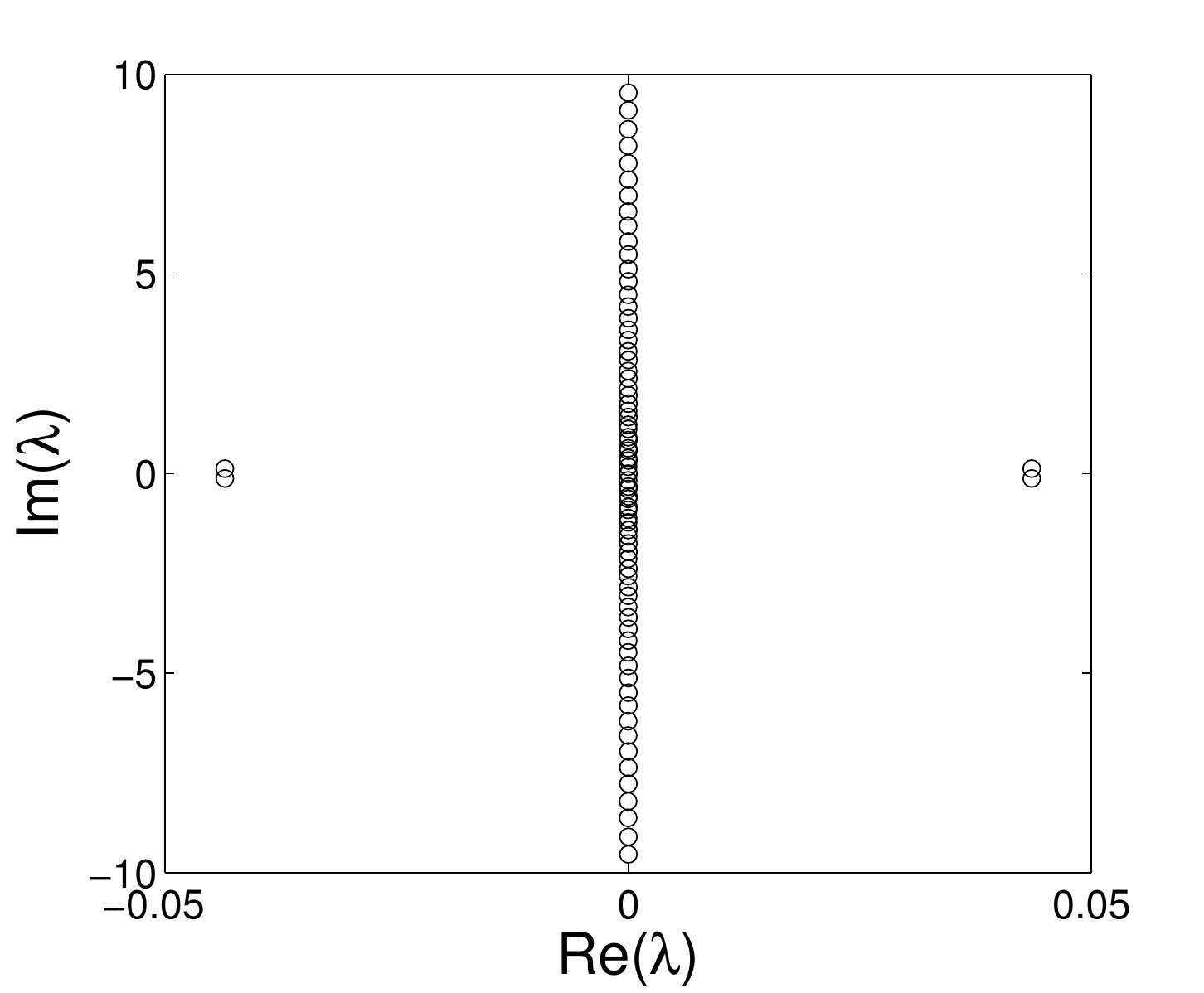}}
\subfigure[]{\includegraphics[scale=.4]{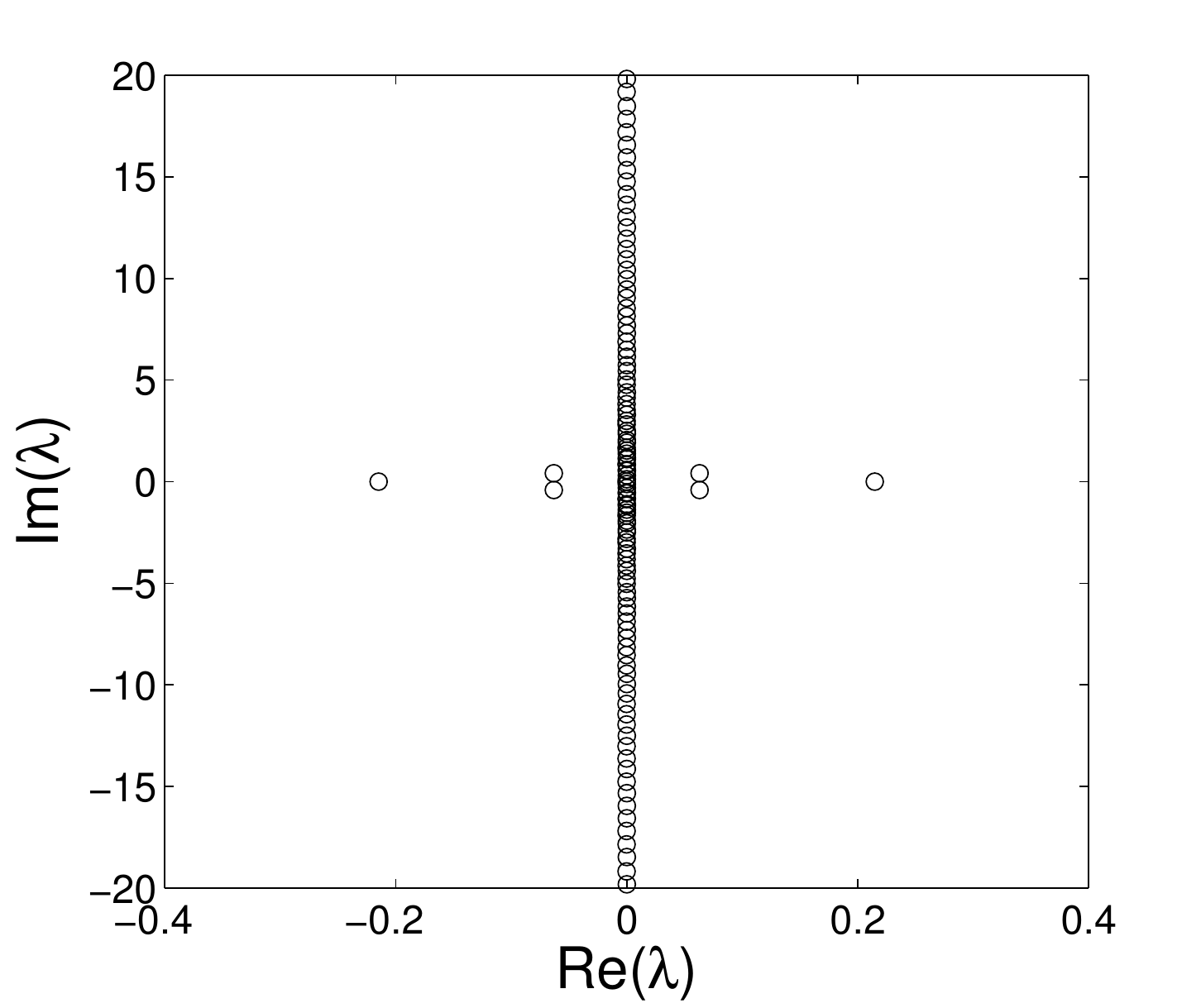}}
\subfigure[]{\includegraphics[scale=.4]{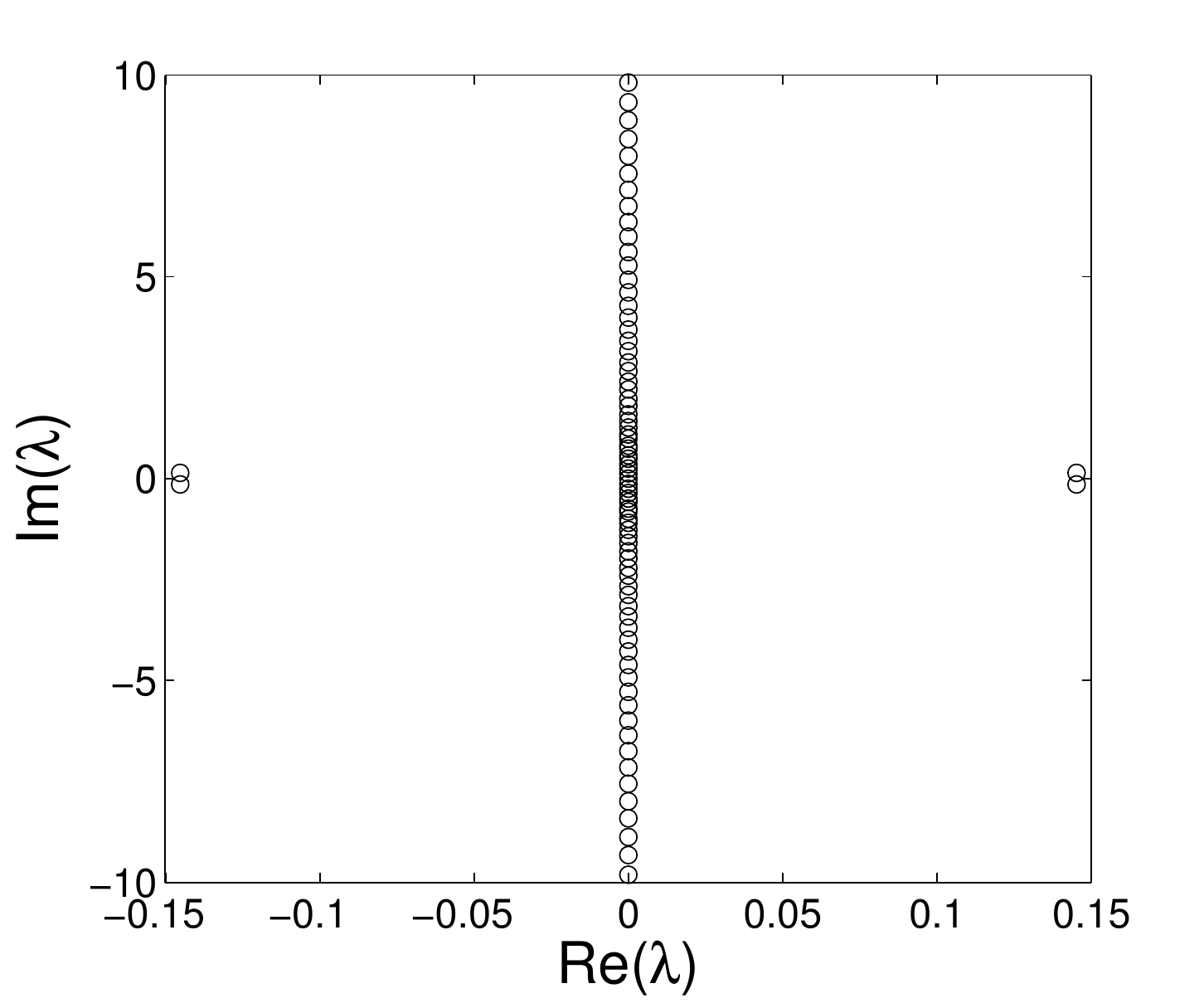}}
\end{center}
\caption{The structures of the spectra of the solutions in Fig.\ \ref{fig:truedark} in the complex plane.}
\label{fig:stab2}
\end{figure}

For the solutions in Fig.\ \ref{fig:truedark}(a)--(d), the numerically computed eigenvalues of the corresponding operator $M$ in the complex plane are shown in Fig.\ \ref{fig:stab2}. Here, one can observe that indeed solutions (a) and (c) are unstable. For both solutions, there is $(|P-Q|=1)-$pair of real eigenvalues, in agreement with Lemma \ref{th:main}. It is interesting to note for solution (c), in addition to the pair of real eigenvalues, there are two pairs of unstable eigenvalues with nonzero imaginary parts. As for solutions (b) and (d), where $P-Q=0$, we found that they are unstable due to the presence of complex unstable eigenvalues, i.e.\ the solitons experience an oscillatory instability. Hence, as opposed to a stability as in the case of positive solutions, $P-Q=0$ corresponds to a quartet of complex eigenvalues here.

By varying the constant $r^2$, we have calculated the stability of the corresponding solutions of all the branches in Fig.\ \ref{fig:truedark}(e), where we found that there is no stable black solitons. Starting from point (a) in the bifurcation diagram, as $r^2$ increases the only positive eigenvalue decreases and disappears in the imaginary axis at the pitchfork bifurcation point. Beginning from an asymmetric solution at point (d), the two pairs of unstable complex eigenvalues also decrease as $r^2$ increases. Increasing $r^2$ further, two pairs of complex eigenvalues immediately appears. The symmetric solutions pick up a pair of real eigenvalues at the saddle-node bifurcation.

The typical time evolution of unstable black solutions is similar to that in the right panel of Fig.\ \ref{fig:evol1}. 

\section{Summary}

We have established analytically an instability criterion for solitary waves on nonzero backgrounds, i.e.\ nontrivial equilibria in the phase-space, through a topological argument. We have applied the method to a nonlinear Schr\"odinger equation with a non-uniform nonlinearity coefficient. In particular, we have investigated a Schr\"odinger equation with self-focusing nonlinearity bounded by self-defocusing one, which admits solutions on nonzero backgrounds. %The present work extended the results of \cite{tran92b,mitc93,{jleon04}}.
Even though we considered a particular system set-up, the novel criterion can be applied generally.

The result presented in this paper is novel and can be effectively used to identify localized solutions possessing real eigenvalues. The method has been successfully applied to predict the instability of both symmetric and asymmetric positive solutions and black solitons. Numerical computations have been used to confirm the strength of the criterion. Through numerics, it also has been shown that the technique cannot definitively establish instability caused solely by complex unstable eigenvalues, such as in asymmetric black solitons. %and symmetric ones for some parameter values%

Even though here we consider the case of the same background as the spatial variable $x\to\pm\infty$, the method can be applied generally to study systems composed of %, e.g., surface dark waves at an interface between
dissimilar defocusing media \cite{skin91,ande91,mira92}. In the papers, Andersen and Skinner \cite{skin91,ande91,mira92} consider the existence and stability of dark surface wave sitting at the interface between two defocusing Kerr media with different linear and nonlinear refractive indices. Depending on the combination of the coefficients, one can obtain solitary waves that decay to a non-zero constant at one side while vanish at the other side, or decay to different non-zero constants at both sides. In those cases, Lemma \ref{th:main} is still valid.

%\begin{acknowledgements}
%If you'd like to thank anyone, place your comments here
%and remove the percent signs.
%\end{acknowledgements}

% BibTeX users please use one of
%\bibliographystyle{spbasic}      % basic style, author-year citations
%\bibliographystyle{spmpsci}      % mathematics and physical sciences
%\bibliographystyle{spphys}       % APS-like style for physics
%\bibliography{}   % name your BibTeX data base
\bibliographystyle{amsalpha}

% Non-BibTeX users please use

\end{document}